\documentclass{article}
\usepackage{arxiv}
\usepackage{tabu}       

\usepackage{nicefrac}       
\usepackage{microtype}      
\usepackage{lipsum}
\usepackage[utf8]{inputenc} 
\usepackage[T1]{fontenc}    

\usepackage{graphicx}  
\usepackage{amsmath}  
\usepackage{hyperref}
\usepackage{multirow}
\usepackage{url}
\usepackage[all]{xy}
\usepackage{microtype}
\usepackage{graphicx}
\usepackage{subcaption}
\usepackage{ upgreek }
\usepackage{booktabs} 
\usepackage{ dsfont }
\usepackage{amsfonts}       
\usepackage{ amssymb }
\usepackage{amsthm}
\usepackage{algorithm}
\usepackage[noend]{algorithmic} 
\usepackage{enumerate}
\usepackage{forloop}
\usepackage{comment}
\usepackage[numbers]{natbib}
\usepackage{float}
\usepackage{tabularx}
\usepackage[nottoc]{tocbibind}
\usepackage{ latexsym }
\usepackage{pst-node}
\usepackage{amssymb,amsthm}
\usepackage{tikz-cd}
\usepackage{mathabx}
\usepackage{esvect}
\usepackage{tikz}
\usepackage{paralist}
\usetikzlibrary{arrows}
\usepackage{nomencl} 
\usepackage{enumitem}
\makenomenclature

\newcommand{\R}{\mathbb{R}}
\newcommand{\N}{\mathbb{N}}
\newcommand{\Z}{\mathbb{Z}}
\newcommand{\field}{\mathds{k}}
\newcommand\SComplex{K}

\newcommand{\Filt}{\mathrm{Filt}}
\renewcommand{\Im}{\mathrm{Im}}

\newcommand{\VertSet}{\mathrm{V}}
\newcommand\simplex{\sigma}
\newcommand\Polytope{\mathcal{P}}

\newcommand\Barc{\mathrm{Bar}}
\newcommand\persmap{\mathrm{PH}}
\newcommand{\FinEnd}{\End^*}
\newcommand{\End}{\mathrm{End}}

\newcommand\Filtration{\mathcal{K}}

\newcommand\I{\text{I}}

\newcommand\EquivClassSimplex{\Phi}

\newcommand\Din{D^{< j}}
\newcommand\Dinbis{D^{< j'}}
\newcommand\pdata{\mathcal{F}}

\newcommand{\dimK}{\ensuremath{\mathrm{dim }K}}
\newcommand{\low}{\ensuremath{\mathrm{low}}}

\newcommand{\True}{\ensuremath{\mathrm{True}}}
\newcommand{\False}{\ensuremath{\mathrm{False}}}

\newcommand{\Results}{\ensuremath{\mathbf{Results}}}

\newcommand{\DFS}{\ensuremath{\mathbf{ExtendFiltration}}}
\newcommand{\AlgoComputeFiltrations}{\ensuremath{\mathbf{ComputeFiltrations}}}

\newcommand{\Intervals}{\mathcal{I}}

\newcommand{\MaxFree}{\ensuremath{\mathrm{M}}}

\newcommand{\TrivCur}{\Phi_{\mathrm{cur}}^{\mathrm{n.c.}}}

\newcommand{\ClassCur}{\Class_{\mathrm{cur}}}
\newcommand{\IntervalsCur}{\Intervals_{\mathrm{cur}}}
\newcommand{\FiltrationIn}{\Filtration_{\mathrm{in}}}
\newcommand{\LowerStarSimplices}{\mathrm{L}}

\newcommand{\Class}{\EquivClassSimplex}

\newcommand{\basis}{E}
\newcommand{\chaincx}{C} 
\newcommand{\filta}{f} 
\newcommand{\filtb}{g} 
\newcommand{\filtersCompat}{\persmap^{-1}(D)} 
\newcommand{\unionofpolytopes}{S}

\newcommand{\filterLowerK}[2]{\mathrm{Low}_{#1}(#2)}

\newcommand{\cell}{\mathcal{P}}
\newcommand{\critvals}{\Gamma} 
\newcommand{\critval}{\gamma} 
\newcommand{\ncritval}{m} 
\newcommand{\birth}{\mathrm{birth}} 
\newcommand{\death}{\mathrm{death}} 

\newcommand{\ambientPoly}{{\mathcal P}}

\newcommand{\prop}{P} 

\newcommand{\cwcomplex}{\SComplex}

\newcommand{\Pairing}{\pi}
\newcommand{\PairingBirth}{\pi_{+}}
\newcommand{\PairingDeath}{\pi_{-}}

\makeatletter
\def\namedlabel#1#2{\begingroup
    #2%
    \def\@currentlabel{#2}%
    \phantomsection\label{#1}\endgroup
}
\newcommand{\rmathbb}[1]{{\mathpalette\dude@rmathbb{#1}}}
\newcommand{\dude@rmathbb}[2]{%
  \scalebox{-1}[1]{$\m@th#1\mathbb{#2}$}%
  }
\makeatother

\theoremstyle{plain}
\newtheorem{theorem}{Theorem}[section]
\newtheorem{lemma}[theorem]{Lemma}
\newtheorem{proposition}[theorem]{Proposition}

\newtheorem{conjecture}{Conjecture}
\newtheorem*{theorem*}{Theorem}
\newtheorem*{proposition*}{Proposition}

\theoremstyle{definition}
\newtheorem{remark}[theorem]{Remark}

\newtheorem{definition}[theorem]{Definition}
\newtheorem{example}[theorem]{Example}

\definecolor{dred}{rgb}{0.7,0.0,0.0}

\begin{document}
	\author{
 Jacob Leygonie
  \\
  Mathematical Institute\\
  University of Oxford\\
  Oxford OX2 6GG, UK \\
  \texttt{jacob.leygonie@maths.ox.ac.uk}
  \And
   Gregory Henselman-Petrusek 
  \\
  Mathematical Institute\\
  University of Oxford\\
  Oxford OX2 6GG, UK \\
  \texttt{henselmanpet@maths.ox.ac.uk} \\
   \\
}
	
	\title{Algorithmic Reconstruction of the Fiber of Persistent Homology on Cell Complexes }

	\maketitle

\begin{abstract}
Let~$\SComplex$ be a finite simplicial, cubical, delta or CW complex. The persistence map~$\persmap$ takes a filter~$f:\SComplex\rightarrow \R$ as input and returns the barcodes~$\persmap(f)$ of the associated sublevel set persistent homology modules. We address the inverse problem: given a target barcode~$D$, computing the fiber~$\persmap^{-1}(D)$. For this, we use the fact that~$\persmap^{-1}(D)$ decomposes as complex of polyhedra when~$\SComplex$ is a simplicial complex, and we generalise this result to arbitrary based chain complexes. We then design and implement a depth first search algorithm that recovers the polyhedra forming the fiber~$\persmap^{-1}(D)$. As an application, we solve a corpus of 120 sample problems, providing a first insight into the statistical structure of these fibers, for general CW complexes.


\end{abstract}

\tableofcontents

\section*{Introduction}
\label{sec_intro}
Persistent Homology (PH) is a computable descriptor~\cite{edelsbrunner2008persistent,zomorodian2005computing} for data science problems where topology is prominent, e.g. for analysing graphs and simplicial complexes~$\SComplex$. Given a function~$f:\SComplex\rightarrow \R$, one assembles the homology groups of the sub complexes~$f^{-1}((-\infty,t])$ into a module which is faithfully represented by a so-called barcode~$D=\persmap(f)$. Using vectorisation methods~\cite{adams2017persistence,bubenik2015statistical} the topological information of the barcode can then be used in statistical studies and machine learning problems. 

To gain an a priori understanding of the problems where~$\persmap$ is applied, it is key to know about the invariance of~$\persmap$, i.e. the context in which it is a discriminating descriptor. Hence the general inverse problem: what are the different functions~$f$ giving rise to the same barcode~$D=\persmap(f)$? Equivalently we are interested in the properties of the fiber~$\persmap^{-1}(D)$ over a target barcode~$D$.  Prior work~\cite{leygonie2021fiber} has shown the fiber to be the geometric realization of a polyhedral complex; each polyhedron represents the restriction of the fiber to strata of the space of filters. Similarly the space of barcodes is given a stratification, and $\persmap^{-1}(D)$ is piecewise-affinely isomorphic to $\persmap^{-1}(D')$ for any~$D, D'$ that belong to the same barcode stratum. Furthermore, if~$D'$ lies in the closure of the stratum containing~$D$, then there is a natural map~$\persmap^{-1}(D) \rightarrow \persmap^{-1}(D')$ that respect the polyhedral structure of~$\persmap^{-1}(D)$ and~$\persmap^{-1}(D')$.

However, understanding the fiber~$\persmap^{-1}(D)$ remains a challenge, in general.  If the underlying simplicial complex~$\SComplex$ is a subdivision of the unit interval or of the circle, then the fiber~$\persmap^{-1}(D)$ is a disjoint union of contractible sets~\cite{cyranka2018contractibility} and circles~\cite{mischaikow2021persistent},  respectively.  Outside these and a few other 1-dimensional examples, however, little is known.  Even establishing that~$\persmap^{-1}(D)\neq \emptyset$ can be difficult, as we shall see in appendix~\ref{appendix:connection_to_zeeman}.\footnote{For certain types of barcode~$D$, we show that $\persmap^{-1}(D) = \emptyset$ iff~$\SComplex$ is collapsible.}  

Algorithm development represents an important step toward addressing this challenge. Computerized calculations offer a range of new examples (intractable by hand) with which to test hypotheses and search for patterns, thus contributing to the growth of theory. They also represent an essential prerequisite to many scientific applications.  

In this paper we propose an algorithmic approach to compute the fiber of the persistence map~$\persmap$, for an arbitrary finite simplicial complex~$\SComplex$, over an arbitrary barcode~$D$.  We then generalize this approach to include finite CW complexes.
%

\paragraph{Outline of contributions} After introducing elementary material from Persistence theory (section~\ref{sec_prelim}), we define in section~\ref{section_compatible_data} the key data structures used in our method. The algorithm, presented in section~\ref{section_computing_fiber}, computes the list of polyhedra in the fiber~$\persmap^{-1}(D)$. For this, we adopt an inductive approach that constructs filtrations of~$K$ simplex by simplex, and simultaneously, an associated barcode via the standard reduction algorithm for computing persistent homology. Through this incremental construction, we ensure that the filtrations remain compatible with the target barcode~$D$. The resulting collection of polyhedra is the polyhedral complex that characterises the homeomorphism type of~$\persmap^{-1}(D)$.

That the fiber~$\persmap^{-1}(D)$ is a polyhedral complex is generalised to arbitrary based chain complexes in section~\ref{section_based_chain_complexes}, in particular including CW complexes, delta complexes, cubical complexes, on which we can take arbitrary filter functions, or the sub spaces of lower-star and lower-edge filter functions. In turn our algorithm adapts to these situations. 

In a section~\ref{section_experiments} dedicated to experiments, we apply the algorithm to multiple complexes and barcodes, and report statistics about the fiber~$\persmap^{-1}(D)$, such as its number of polyhedra and its Betti numbers. Sometimes unexpectedly, most of the properties observed for the $1$-dimensional complexes studied in~\cite{cyranka2018contractibility, leygonie2021fiber,mischaikow2021persistent} do not hold for more general complexes. For instance, already for a triangulated $2$-sphere the fiber~$\persmap^{-1}(D)$ over some barcodes~$D$ has non-trivial homology in degree~$3$, unlike all the existing examples of graphs, for which~$\persmap^{-1}(D)$ has vanishing homology in dimension greater than~$1$. We also find CW complexes, for example the real projective plane, that are topological manifolds whose fibers are not necessarily manifolds. 

By contrast, the algorithm allows us to observe novel trends that hold consistently across examples: for instance to each simplicial complex~$\SComplex$ we associate a specific barcode~$D_{\SComplex}$ such that the fiber~$\persmap^{-1}(D_{\SComplex})$ and~$\SComplex$ have the same Betti numbers (see conjecture~\ref{conjecture_fiber_homotopy_equivalent_to_K}). Our code base has private dependencies, hence will be made public at a later time.

\paragraph{Related works} Finally we note that for related inverse problems algorithmic approaches to reconstruct the fiber have been designed. For instance instead of taking a single function the Persistent Homology Transform (PHT) of~\cite{turner2014persistent} computes the barcodes of a family of functions over a fixed shape in~$\R^3$, which is enough to completely characterise this shape, see~\cite{curry2018many,ghrist2018persistent} for generalisations to higher dimensions. These sorts of results motivated the design of algorithms to reconstruct the complex~$\SComplex$ from the associated family of barcodes~\cite{betthauser2018topological,fasy2019persistence}.
\paragraph{Acknowledgments} Both authors thank Ulrike Tillmann and Heather Harrington for close guidance and numerous interactions. Both authors acknowledge the support of the Centre for Topological Data Analysis of Oxford, EPSRC grant EP/R018472/1. J.L.'s research is also funded by ESPRC grant EP/R018472/1.  G. H.-P. acknowledges support from NSF grant DMS-1854748.
%

\section{Background}
\label{sec_prelim}
We fix a finite simplicial complex~$\SComplex$ of dimension~$\dimK$ and  with $\sharp K$ many simplices. Throughout, (simplicial) homology is taken with coefficients in a fixed field $\field$. A {\em filter} over~$\SComplex$ is a map~$f:K\rightarrow \I$ valued in the interval $\I:=[1;\sharp \SComplex ] \subset \mathbb R$, whose sublevel sets are subcomplexes of~$\SComplex$. If we regard $\SComplex$ as a poset partially ordered by face inclusions~$\simplex \subseteq \simplex'$, then~$f$ can be regarded, equivalently, as an order-preserving function into~$\I$. The set of filters is a polyhedron contained in the Euclidean space~$\R^K$.

For each homology degree~$0\leqslant p \leqslant \dimK$, we have a (finite) persistent homology module arising from the sub-level sets of~$f$. The isomorphism type of each module of this form is uniquely determined by the associated barcode~$\persmap_p(f)$; concretely, the barcode is a finite multi-set of pairs~$(b,d)$, called {\em intervals}, each of which characterizes the appearance and annihilation of a class of~$p$-cycles in the filtration.  For further details, see \cite{zomorodian2005computing, EHComputational10}. We denote the set of all barcodes by~$\Barc$.

The persistence map~$\persmap$ takes a filter~$f$ and returns the $\dimK+1$ barcodes of interest:
\[\persmap: f\in \Filt(\SComplex) \longmapsto \big(\persmap_0(f),\persmap_1(f),\cdots, \persmap_{\dimK}(f)\big)\in \Barc^{\dimK+1}.\]
Our goal is to compute the fiber~$\persmap^{-1}(D)$ over some $(d+1)$-tuple of barcodes~$D=(D_0,D_1,\cdots,D_{\dimK})\in \Barc^{\dimK+1}$, from now on called barcode for simplicity.

There is another, equivalent formulation of~$D$ which is sometimes better suited to formal arguments. 
In this approach, we replace the sequence of multisets~$D$ with a single, bona fide set~$X$.  We regard~$X$ as a set of indices~$x$, with one index for each interval in the barcode, and write~$\birth(x)$,~$\death(x)$, and~$\dim(x)$ for the birth time, death time, and dimension, respectively, of the corresponding interval.  
Thus~$D_p$ is the multiset $\{(\birth(x), \death(x) ) : x \in X, \; \dim(x) = p \}$.  The set~$X$ is especially useful for algorithms, e.g.\ when writing for-loops, however it is nonstandard as a convention.%

By way of compromise, we will identify~$D$ with~$X$.  When we refer to fixing~$(b,d) \in D$, we mean fixing~$x \in X$ such that~$\birth(x) = b$ and~$\death(x) = d$.  By~$\dim(b,d)$, we then mean~$\dim(x)$.

We also abuse notation and write~$(b,d)\in D$ whenever~$(b,d)\in D_p$ for some~$0\leq p\leq \dimK$, and we set~$\dim(b,d):=p$.  

Given a barcode~$D$, let
\[
        \End(D) =  \{ b :  (b,d) \in D \text{ for some } d \} \cup  \{ d :  (b,d) \in D \text{ for some } b \} \subseteq \R \sqcup \{+\infty\}
\]
denote the set of all endpoints of intervals in~$D$. Set  
$\FinEnd(D) = \End(D) \backslash\{\infty\}$
and write~$\dim D = |\FinEnd(D)|$ for the number of finite endpoints.  In particular,~$\dim D \leq \sharp \SComplex$. By normalizing, we may assume without loss of generality that $\FinEnd(D)$ has form $\{1, \ldots, \dim D\}$.  For instance, the barcode~$D$ with~$4$ intervals (in possibly different homology degrees)~$(1,2)$,~$(1,3)$,~$(2,3)$ and~$(2,+\infty)$ has dimension~$3$, since the endpoints values form the set~$\{1,2,3\}$. 

For any pair of filters~$\filta, \filtb: \basis \to \I$, let us define the relation~$\filta \sim \filtb$ by either of the following two equivalent axioms:
    \begin{enumerate}[label=\textbf{(A\arabic*)}]
        \item 
        \label{item:cellaxiom_compose_background}        
        There exists an order-preserving map $\psi: \I \to \I$ such that $\psi(i) = i$ for each~$i \in \{ 1, \ldots,  \dim D\}$, and $\filtb = \psi \circ \filta$.    
        \item 
            \label{item:cellaxiom_inequality_background}
            The function $\filtb$ satisfies
            \begin{align}
                \filta(e) \in \FinEnd(D) \implies \filtb(e) = \filta(e)
                &&
                \filta(e) \le \filta(e') \implies \filtb(e) \le \filtb(e').
                \label{eq:equivalence_condition_background}
            \end{align}        
    \end{enumerate}
Note that~$\sim$ is reflexive and transitive but not symmetric. Given a filter~$\filta$ we write~$\eta^i := \filta(\basis) \cap (i, i+1)$ and 
\[
\cell(\filta) := \{ \filtb \in \I ^E  : \, \, \filta \sim \filtb \} \subseteq \I^{\SComplex}.
\]
Theorems~\ref{theorem_polyhedral_complex_a} and~\ref{theorem_polyhedral_complex} below are proved in~\cite[Theorem~2.2]{leygonie2021fiber}.  Recall that a finite set~$\ambientPoly$ of polyhedra in~$\R^\SComplex$ is a \emph{polyhedral complex} if for each polyhedron~$X\in \ambientPoly$, the faces of~$X$ belong to~$\ambientPoly$ as well, and if, furthermore, any two polyhedra~$X,X'\in \ambientPoly$ intersect at a common face. The underlying space of~$\ambientPoly$ is~$\vert \ambientPoly \vert =  \bigcup_{X \in \ambientPoly} X$.
\begin{theorem}
\label{theorem_polyhedral_complex_a}
Let~$\filta\in \persmap^{-1}(D)$ be a filter in the fiber. Then~$\cell(\filta)\subseteq \persmap^{-1}(D)$, and~$\cell(\filta)$ is a polyhedron which is affinely isomorphic to the product~$\Delta_{\# \eta^1} \times \cdots \times \Delta_{\# \eta^{\dim D}}$, where~$\Delta_{k}$ stands for the standard geometric simplex of dimension~$k$.
\end{theorem}
The polyhedra~$\cell(\filta)$ then assemble into a polyhedral complex with underlying space the fiber~$\persmap^{-1}(D)$:
\begin{theorem}
\label{theorem_polyhedral_complex}
The set~$\big\{ \cell(\filta) \mid \filta\in \persmap^{-1}(D) \big\}$ is a polyhedral complex. 
\end{theorem}
In section~\ref{section_based_chain_complexes} we extend Theorems~\ref{theorem_polyhedral_complex_a} and~\ref{theorem_polyhedral_complex} to CW complexes and more generally to based chain complexes. 
\section{Data structures}
\label{section_compatible_data}
\subsection{Classifications of simplices compatible with a barcode}
\label{section_compatible_filtrations}
Let us fix a finite simplicial complex~$\SComplex$ and a barcode~$D$. 

An important part of our strategy will be to introduce the concept of~$D$-compatible factorization. This object allows one to define clean correspondences between simplices in~$\SComplex$ and interval endpoints. This, in turn, allows one to define equivalence classes of simplices according to whether or not they contribute to the barcode~$D$ -- a major benefit, in terms of overall organization. 

Let $f\in \persmap^{-1}(D)$ be given. Then $\FinEnd(D) = \{1,\cdots, \dim D\}\subseteq \mathrm{Im}(f)$. Note that containment may be proper because $f$ can take non-integer values, in general.  However, we claim  $\{1,\cdots, \dim D\} = \mathrm{Im}(f)$ when $f$ is injective; in this case simplices appear ``one at a time'' in the corresponding filtration, and a simple counting argument allows one to infer that $\dim D = \# \SComplex = \#\mathrm{Im}(f)$ as a result.  We make use of this property in the general case.

\begin{figure}[ht]
\centering
\includegraphics[width=0.7\textwidth]{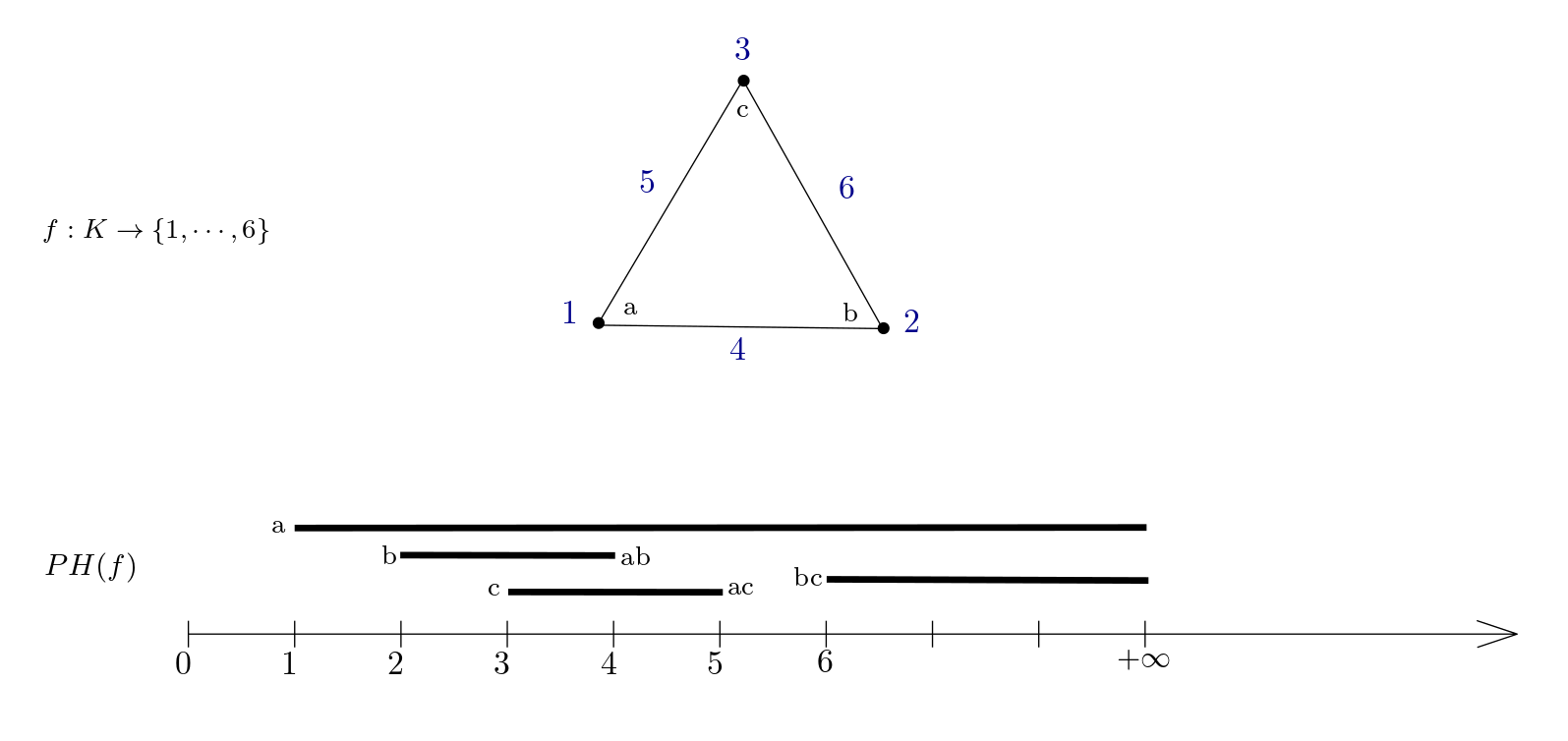}
\caption{{\footnotesize Given an injective filter~$f$, the finite  endpoints in the barcode~$\persmap(f)$ are canonically identified with simplices in~$\SComplex$. Here,~$\SComplex$ is the triangle with vertices denoted by a, b and c. The values (in dark blue) of the bijection $f:\SComplex \rightarrow \{1,\cdots, 6\}$ are displayed along the simplices of~$\SComplex$. No two distinct intervals in the resulting barcode (all homology degrees included) share a common end value, and these endpoints are annotated with the corresponding simplex}}
\label{fig:injective}
\end{figure}
\begin{definition}
\label{definition_compatible_injective_function}
Let~$f$ be an injective filter valued in~$\{1,\cdots, \sharp \SComplex\}$. The  \emph{associated total order} is the  unique total order on simplices~$\simplex_1<\cdots <\simplex_{\sharp \SComplex}$ such that  $f(\simplex_k) = k$ for each $k$. If there exists a non-decreasing continuous map~$\phi: \R \mapsto \I$ such that $\persmap(\phi\circ f)=D$, we say that~$(f,\phi)$ is a {\em $D$-compatible factorization}.
\end{definition}
\begin{remark}
\label{remark_equivariance}
Given a non-decreasing continuous map~$\phi$ and an arbitrary barcode~$D' = [D'_0, \ldots, D'_{\# \SComplex}]$, let $\phi(D')$ be the barcode such that $\phi(D')_p$ contains one interval of form~$(\phi(b'),\phi(d'))$ for each interval~$(b',d')$ in~$D'_p$ that satisfies~$\phi(b')\neq \phi(d')$. Here, by convention, we define $\phi(\infty) = \infty$. Then the persistence map is equivariant with respect to this action, in the sense that~$\persmap(\phi\circ f)= \phi(\persmap(f))$; see for instance~\cite[Lemma~1.5]{leygonie2021fiber}. This gives another perspective on the idea of $D$-compatibility: The barcode~$\persmap(f)$ has endpoints in bijection with the simplices of~$\SComplex$, and~$\phi$ \lq contracts' this barcode into~$D$. See Fig.~\ref{fig:nailing_and_free} for illustration.
\end{remark}

We now describe how a~$D$-compatible factorization~$(f,\phi)$ allows one to disambiguate the homological roles of each simplex. Let~$\simplex\in K$ be a given. By injectivity of~$f$, there exists a unique~$\simplex'\in K$ such that $(f(\simplex), f(\simplex'))$ or  $(f(\simplex'), f(\simplex))$ lies in the barcode $\persmap(f)$. We write $[\simplex,\simplex']_f$ for this pairing; where context leaves no room for confusion, we simply write~$[\simplex,\simplex']$. There is also the case where~$\simplex$ is unpaired, i.e. when it corresponds to an infinite interval~$(f(\simplex), \infty)$. We then write~$[\simplex ,\infty]$.
Then~$(f,\phi)$ gives rise to the following classification of pairs of simplices~$[\simplex,\simplex']$ (with the possibility that~$\simplex'=\infty$) according to whether or not~$(\phi\circ f(\simplex),\phi\circ f(\simplex'))$ is an interval in the barcode~$D$:
\begin{description}
\item[\namedlabel{itm:resp}{Critical}]~$(\phi\circ f(\simplex),\phi\circ f(\simplex'))=(b,d)$ for some non-trivial interval~$(b,d)\in D$; We say that~$[\simplex,\simplex']$ is {$(f,\phi)$-\em critical} for~$(b,d)$ and record this association via~$\PairingBirth(b,d):=\simplex$ and~$\PairingDeath(b,d):=\simplex'$;
\item[\namedlabel{itm:triv}{Non-Critical}] Otherwise~$\phi$ cancels the interval $(f(\simplex),f(\simplex'))$, i.e. $\phi\circ f(\simplex)=\phi\circ f(\simplex')$; We say that~$[\simplex,\simplex']$ is {\em $(f,\phi)$-non-critical}.
\end{description}

In particular we have a well-defined bijection of intervals in~$D$ with critical pairs:
\[\Pairing: D \to \SComplex \times (\SComplex\sqcup \{\infty\}), \quad (b,d) \longmapsto [\PairingBirth(b,d), \PairingDeath(b,d)]. \]%
By abuse of terminology we say that a simplex~$\simplex$ is {\em critical} if it is part of a critical pair, and that it is {\em non-critical} otherwise. Note that an unpaired simplex is necessarily critical for some~$(b,\infty)\in D$. In Fig.~\ref{fig:nailing_and_free}, we provide an example of $D$-compatible~$(f,\phi)$ and the resulting critical and non-critical simplices. 

There are uncountably infinitely many $D$-compatible factorizations. In the two following definitions we compress the necessary information into finitely many equivalence classes of $D$-compatible factorizations. 
\begin{definition}
\label{definition_filtration}
A {\em classification}~$\Filtration$ of simplices of~$\SComplex$ relative to~$D$, or {\em classification} for short, consists of:
\begin{itemize}
    \item[(i)] A bijective filter~$f:\SComplex\rightarrow \{1,\cdots, \sharp \SComplex\}$, or equivalently a total ordering~$\simplex_1<\cdots <\simplex_{\sharp \SComplex}$;
    \item[(ii)] A partition of simplices into consecutive ordered sets~$\Class_1,\cdots,\Class_m$, which we refer to as~{\em classes}:
      \begin{align*}
        \underbrace
            {\simplex_1<\simplex_2<\cdots<\simplex_{\sharp \Class_1}}_
            { \Class_1 }
        < 
        \underbrace
            {\cdots < \cdots < \cdots }_
            {\Class_2, \cdots , \Class_{m-1} } 
       <           
        \underbrace
            {\simplex_{(\sharp \SComplex- \sharp \Class_m +1)}<\cdots < \simplex_{\sharp \SComplex}}_
            {\Class_m };        
    \end{align*}
    \item[(iii)] For each interval endpoint~$j$, a class~$\Class_{i_j}$, with $j \longmapsto i_j$ a non-decreasing assignment.
\end{itemize}
\end{definition}
\begin{definition}
\label{definition_free_nail_compatibledata}
A $D$-compatible classification~$\Filtration$ is a classification induced by a~$D$-compatible factorization~$(f,\phi)$ as follows:
\begin{itemize}
    \item[(i)] $f:\SComplex\rightarrow \{1,\cdots, \sharp \SComplex\}$ equals the bijective filter of~$\Filtration$ inducing the order~$\simplex_1<\cdots <\simplex_{\sharp \SComplex}$; 
    \item[(ii)] The image of~$\phi\circ f$ has cardinality~$m$, i.e. we can write $\Im(\phi\circ f)=\{t_1<\cdots<t_m\}$, and the associated sequence of pre-images equal the classes of~$\Filtration$:
     \[(\phi\circ f)^{-1}(t_1)=\Class_1,\, \,  \cdots, \, \, (\phi\circ f)^{-1}(t_m)=\Class_m. \]
     \item[(iii)] For each interval endpoint~$j \in \FinEnd(D)$, we have $(\phi\circ f)^{-1}(j)=\Class_{i_j}$. 
\end{itemize}
\end{definition}
\begin{figure}[ht]
\centering
\includegraphics[width=0.9\textwidth]{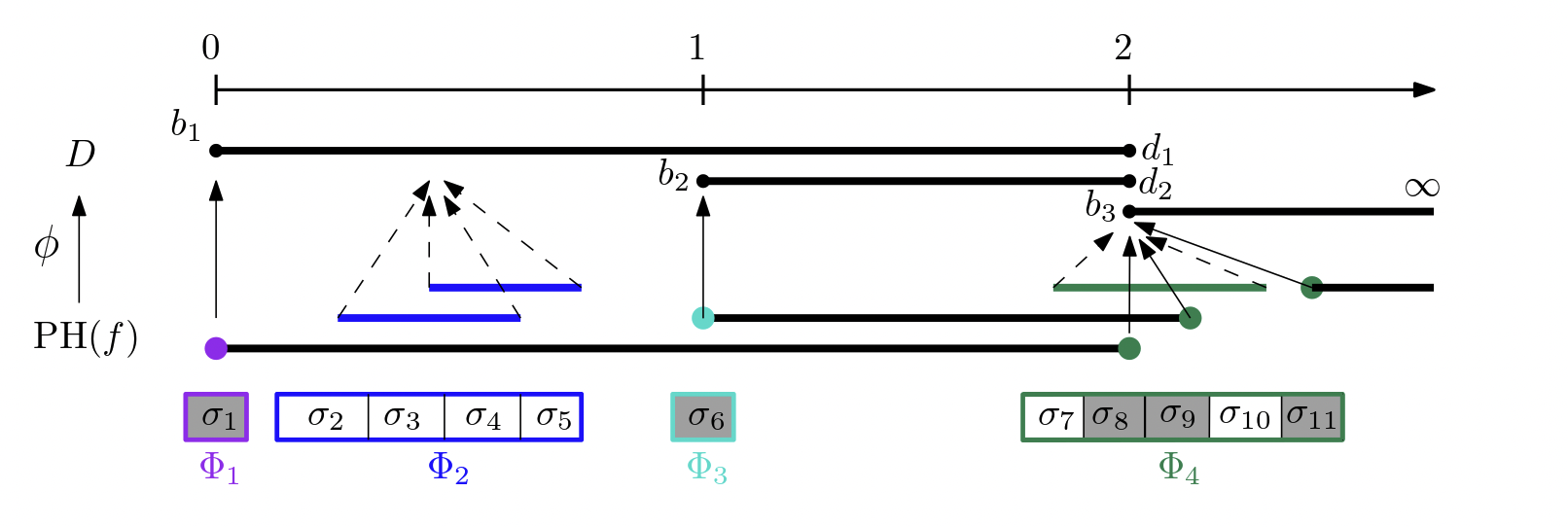}
\caption{ With a simplicial complex$~\SComplex$ and an injective filter~$f$ in the background, $\phi$ maps the  barcode~$\persmap(f)$ to~$D$ via the arrows going up, so that~$(f,\phi)$ is $D$-compatible. In the corresponding~$D$-compatible classification, there are~$4$ distinct equivalence classes $\Class_i$ of simplices, and they are distinguished by the use of different colors. Critical simplices are grayed and correspond to intervals via the upward solid arrows: Here~$\Pairing(b_1,d_1)=[\simplex_1,\simplex_8]$,~$\Pairing(b_2,d_2)=[\simplex_6,\simplex_9]$ and~$\Pairing(b_3,\infty)=[\simplex_{11},\infty]$, hence the critical simplices form the set~$\{\simplex_1,\simplex_6,\simplex_8,\simplex_9,\simplex_{11}\}$. %
The other simplices are non-critical and cancelled in pairs by the dotted arrows:~$\{\simplex_2,\simplex_3,\simplex_4,\simplex_5,\simplex_7,\simplex_{10}\}$.%
}
\label{fig:nailing_and_free}
\end{figure}
We depict all this data in Fig~\ref{fig:nailing_and_free}. Note that given a classification~$\Filtration$ the filter~$f$ induces pairings~$[\simplex,\simplex']$.  
\begin{proposition}
\label{proposition_filtration_is_D_compatible}
A classification~$\Filtration$ is the~$D$-compatible classification induced by some~$D$-compatible factorization if and only if the following two rules are satisfied:
\begin{description}
    \item[\namedlabel{RespRule}{Critical rule}:] For any~$1\leq j<j'\leq\dim D$ and~$0\leq p\leq \sharp \SComplex$, there are as many pairs~$[\simplex_j, \simplex_{j'}]_f$ with~$\simplex_j\in \Class_{i_j}$ and~$\simplex_{j'}\in \Class_{i_{j'}}$, here $\dim \simplex_j=p$, as there are copies of the interval~$(j,j')$ in~$D_p$;
    \item[\namedlabel{TrivRule}{Non-critical rule}:] For any remaining pair~$[\simplex ,\simplex']$, both~$\simplex$ and~$\simplex'$ belong to the same class~$\Class_i$.
\end{description}
\end{proposition}
\begin{proof}
Indeed, when these two rules are satisfied, we can define~$\phi$ directly on the classes~$\Class_1,\cdots,\Class_m$ as any order-preserving map sending~$f(\Class_{i_j})$ to~$j$. 
\end{proof}
%
It is natural to group~$D$-compatible factorizations~$(f,\phi)$ according to the classification~$\Filtration$ they induce: it is clear that whenever factorizations induce the same classification, then they induce the same pairs~$[\simplex,\simplex]$ of simplices, the same critical and non-critical pairs (and simplices) and the same bijection~$\Pairing$ from intervals in~$D$ to critical pairs. Therefore these concepts are defined as well given a~$D$-compatible classification~$\Filtration$. 
\subsection{Relations to the fiber}
\label{section_retrieving_polyhedra}

We explain how to retrieve the polyhedra that compose the fiber~$\persmap^{-1}(D)$ (see Theorem~\ref{theorem_polyhedral_complex}) from the set of~$D$-compatible classifications~$\Filtration$. Note that if two~$D$-compatible factorizations~$(f,\phi)$ and~$(f',\phi')$ induce the same classification~$\Filtration$, then by Axiom \ref{item:cellaxiom_inequality_background}, they determine the same polyhedron,~$\cell(\phi \circ f)=\cell(\phi' \circ f')$. Thus the following definition is unambiguous:
\[\cell(\Filtration):=\cell(\phi\circ \filta).\]
Moreover, $\cell(\Filtration)\subseteq \persmap^{-1}(D)$ by Theorem \ref{theorem_polyhedral_complex_a},  because~$(f,\phi)$ is $D$-compatible. Conversely, a filter~$\filtb \in \persmap^{-1}(D)$ can always be written as~$\filtb=\phi\circ \filta$ for some injective filter~$\filta$ and non-decreasing map~$\phi: \I \rightarrow \I$, so that~$\cell(\filtb)=\cell(\phi\circ \filta)$. Therefore we have the following result:
\begin{theorem}
\label{theorem_description_polytopes}
The set~$\big \{ \cell(\Filtration) \mid \Filtration \text{ is a $D$-compatible classification} \big \}$ is the polyhedral complex underlying~$\persmap^{-1}(D)$.  
\end{theorem}
The upshot is that it is enough to compute the $D$-compatible classifications~$\Filtration$ in order to cover the fiber~$\persmap^{-1}(D)$. 
\section{Algorithm}
\label{section_computing_fiber}
We now propose an algorithm to retrieve the fiber~$\persmap^{-1}(D)$, i.e. that computes all the $D$-compatible classifications~$\Filtration$ of the previous section. Our implementation is simple in spirit: Algorithm~\ref{alg:DFS_algo_compatible_data} builds these classifications from scratch, simplex by simplex, and tries all the possible ways to do so. Along the way the algorithm manipulates {\em partial}~$D$-compatible classifications, which can be thought of as the result of cutting a~$D$-compatible classification~$\Filtration$ at a given step (see~Fig.~\ref{fig:partial_filt}).

Given an integer~$1\leq j\leq \dim D+1$, let~$\Din=(\Din_0,\cdots, \Din_{\dimK})$ be the truncated version of~$D$ such that 
\begin{align*}
    \Din_p :=   \bigg\{(b,d) \in D_p : d < j \bigg\} \cup \bigg\{ (b, \infty) \in D_p : (b,d) \in D_p, \; \; b< j \leq d \bigg\}.
\end{align*}
for each $0 \le p \le \dimK$.
\begin{definition}
\label{definition_partial_filtration}
A {\em partial classification} is a tuple 
\[\pdata=\big( \FiltrationIn, j, \ClassCur, \IntervalsCur,\TrivCur)\]
where~$\FiltrationIn$ is a~$\Din$-compatible classification, and the set~$\ClassCur$ is a linearly ordered subset of $\SComplex \backslash \FiltrationIn$ called the \emph{current class}. We say that the current class~$\ClassCur$ is {\em incomplete} under either of the following two non-exclusive conditions:

\begin{center}
 ({\bf \ref{RespRule} violation})~$\IntervalsCur\neq \emptyset$;   \hspace{2mm}    and   \hspace{2mm}   ({\bf \ref{TrivRule} violation})~$\TrivCur\neq \emptyset$.
\end{center}
Otherwise~$\ClassCur$ is {\em complete}. The ({\bf \ref{RespRule} violation}) occurs when the current class~$\ClassCur$ is a truncated version of the critical class~$\Class_{i_j}$ containing the simplices critical for the interval endpoint~$j$, and there are still intervals~$(b,j)$ or~$(j,d)$ that are unpaired with simplices of~$\SComplex$. These missing intervals are stored in~$\IntervalsCur$. Meanwhile, the ({\bf \ref{TrivRule} violation}) indicates that some non-critical birth simplices~$\tau \in \ClassCur$ are not yet paired with a death non-critical simplex, i.e.~$\tau$ creates a $\dim \tau$-cycle which must be destroyed in the same class. These unpaired non-critical simplices are stored in~$\TrivCur$. 
\end{definition}

\begin{figure}[ht]
\centering
\includegraphics[width=0.7\textwidth]{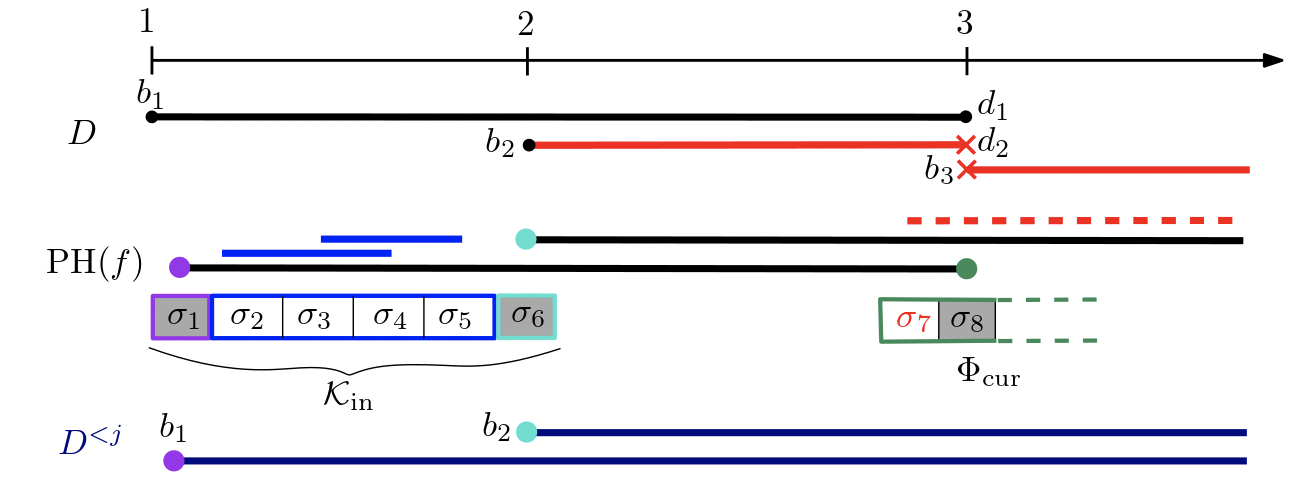}
\caption{ {\footnotesize We consider the barcode~$D$ (top) of Fig.~\ref{fig:nailing_and_free} and cut the $D$-compatible classification at~$\simplex_8$ to get a partial classification. The first three classes~$\Phi_1,\Phi_2,\Phi_3$ are fully present in the resulting partial classification~$\pdata$. So here~$\FiltrationIn$ is a~$\Din$-compatible classification where~$\Din$ (bottom) is the barcode spanned by the interval endpoints~$1$ and~$2$. The next endpoint value to consider is then~$j=3$. The current class~$\ClassCur=\{\simplex_7,\simplex_8\}$ is incomplete for two reasons. On the one hand~{\bf \ref{RespRule}} is violated: The current class~$\ClassCur$ is intended to contain critical simplices for the endpoint~$j=3$, as already it contains the critical simplex~$\simplex_8=\PairingDeath(b_1,d_1)$, but the intervals~$(b_2,d_2)$ and~$(b_3,\infty)$ (in red) also have the endpoint~$j=3$ and are yet unmatched: so~$\IntervalsCur=\{(b_2,d_2),(b_3,\infty)\}$. Hence we should define the simplex~$\PairingDeath(b_2,d_2)$ that destroys the cycle generated by~$\simplex_6=\PairingBirth(b_2,d_2)$, and a simplex~$\PairingDeath(b_3,\infty)$ that creates a~$\dim (b_3,\infty)$-cycle, before completion of the class. On the other hand~{\bf \ref{TrivRule}} is violated: here~$\TrivCur=\{\simplex_7\}$ contains the unpaired simplex~$\simplex_7$ which generates an unexpected (dotted red) interval, which shall be destroyed by another simplex before completion of the class.}}
\label{fig:partial_filt}
\end{figure}
Algorithm~\ref{alg:DFS_algo_compatible_data} builds all the partial classifications starting with the empty one:
\[\pdata=\big( \FiltrationIn, j, \ClassCur, \IntervalsCur,\TrivCur):= (\emptyset,1,\emptyset,D^{<2},\emptyset),\] 
and records complete $D$-compatible classifications in a list~$\Results$. Note that the set~$\IntervalsCur$ of intervals is initialised with intervals starting at the first endpoint~$1$ of~$D$, since it is necessary that the first simplex to enter the filtration will be critical for such an interval. To check that~$\Filtration{}$ is complete is done at line 1 of Algorithm \ref{alg:DFS_algo_compatible_data}, and means that all the simplices of~$\SComplex$ have entered the filtration and that the target barcode~$D$ has been reached. 

If the partial classification~$\pdata$ is not complete but the current class~$\ClassCur$ is complete, i.e. the algorithm checks that~{\bf \ref{RespRule}} and~{\bf \ref{TrivRule}} are satisfied (line 3), then~$\ClassCur$ is added to the classification (line 4), and the algorithm prepares the next class to build according to the following alternative. Either the next class will be the class~$\ClassCur=\Class_{i_{j}}$ that will contain all simplices critical for intervals that have the endpoint~$j$ (lines~6 to~9), or the next class will contain only non-critical simplices (line~10). In practice either we fill~$\IntervalsCur$ with all intervals~$(b,j)$ and~$(j,d)$ in~$D$ that have~$j$ as endpoint, or we set~$\IntervalsCur=\emptyset$.

The remainder of the algorithm enumerates all ways to extend the partial classification (which is provided by the user as input) by adding one simplex to the current class~$\ClassCur$. There are four possible types of extensions, which we explain and illustrate with the example depicted in Fig.~\ref{fig:currentclass}.
\begin{figure}[ht]
\centering
\includegraphics[width=0.5\textwidth]{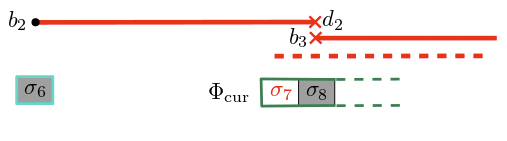}
\caption{ {\footnotesize We consider the part of the partial classification from Fig.~\ref{fig:partial_filt} that is relevant to the current class~$\ClassCur$ only. For~$\ClassCur$ to be completed we must at least find critical simplices~$\PairingBirth(b_3,\infty)$ and~$\PairingDeath(b_2,d_2)$ to account for the red intervals~$(b_2,d_2),(b_3,\infty)$, and pair the non-critical simplex~$\simplex_7$ with another non-critical simplex to destroy the unexpected dotted red interval. }}
\label{fig:currentclass}
\end{figure}
\begin{figure}[ht]
     \centering
     \begin{subfigure}[b]{0.4\textwidth}
         \centering
         \includegraphics[width=\textwidth]{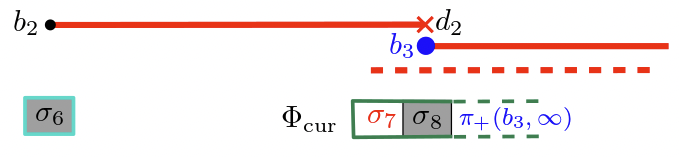}
         \caption{Adding a birth simplex~$\PairingBirth(b_3,\infty)$ generating a~$\dim(b_3,\infty)$-cycle, to be critical for the interval~$(b_3,\infty)$. }
         \label{fig:birth_extension}
     \end{subfigure}
     \hfill
     \begin{subfigure}[b]{0.4\textwidth}
         \centering
         \includegraphics[width=\textwidth]{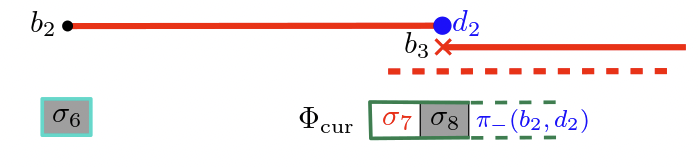}
         \caption{Adding a death simplex~$\PairingDeath(b_2,d_2)$ destroying the~$\dim(b_2,d_2)$-cycle generated by~$\PairingBirth(b_2,d_2)$, to be critical for the interval~$(b_2,d_2)$.}
         \label{fig:death_extension}
     \end{subfigure}
     \begin{subfigure}[b]{0.4\textwidth}
         \centering
         \includegraphics[width=\textwidth]{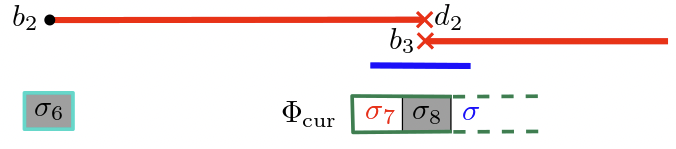}
         \caption{Finding a simplex~$\simplex$ to destroy the cycle generated by an unpaired non-critical simplex~$\simplex_7\in \TrivCur$, resulting in the removal of~$\simplex_7$ in~$\TrivCur$.}
         \label{fig:trivial_death_extension}
     \end{subfigure}
     \hfill
     \begin{subfigure}[b]{0.4\textwidth}
         \centering
         \includegraphics[width=\textwidth]{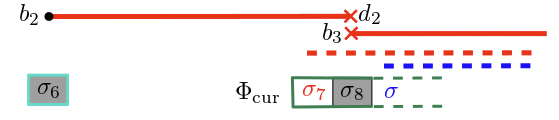}
         \caption{Finding a non-critical simplex~$\simplex$ generating a new cycle, resulting in the addition of~$\simplex$ in~$\TrivCur$.}
         \label{fig:trivial_birth_extension}
     \end{subfigure}
    \caption{The four types of extensions of~$\ClassCur$ (highlighted in blue).}
        \label{fig:extensions}
\end{figure}
\begin{itemize}
    \item(Lines~14 to~19) The first type of extension consists in adding a simplex~$\PairingBirth(j,d)$ critical for the birth of an interval~$(j,d)\in \IntervalsCur$ as in Fig.~\ref{fig:birth_extension}.
    \item(Lines~20 to~26) The second type of extension consists in adding a critical simplex~$\PairingDeath(b,j)$ to account for the death of an interval~$(b,j)\in \IntervalsCur$ as in Fig.~\ref{fig:death_extension}.
    \item (Lines~27 to~31) The third type of extension consists in adding a simplex~$\simplex$ to pair with a non-critical unpaired simplex~$\simplex'\in \TrivCur$ as in Fig.~\ref{fig:trivial_death_extension}.
    \item (Lines~31 to~36) The fourth type of extension consists in adding a non-critical simplex~$\simplex$ to~$\TrivCur$ as in Fig.~\ref{fig:trivial_birth_extension}, i.e. a simplex that creates a $\dim \simplex$-cycle that shall later be destroyed. 
\end{itemize}

To figure out which simplices~$\simplex\notin \FiltrationIn\cup \ClassCur$ can extend the partial~$D$-compatible classifications~$\pdata$ through the algorithm, we maintain a matrix~$\delta$ which is derived from the boundary matrix~$\partial$ via elementary row and column operations. The range of possible extensions can be deduced from the sparsity pattern of $\delta$ as follows;  see~\cite{cohen2006vines} and Remark~\ref{remark_matrix_reduction} for full details on the construction and interpretation of this matrix.
\begin{itemize}
    \item Simplex~$\simplex$ can be added as a (critical or non-critical) simplex to~$\ClassCur$ if and only if~$\delta \simplex=0$ or the lowest non-zero entry of column associated to~$\simplex$ corresponds to another simplex~$\simplex'$ that already belongs to the partial classification:~$\simplex'\in \FiltrationIn\cup \ClassCur$ . We then write~$\low (\delta \simplex):=\simplex'$;
    \item If a simplex~$\simplex$ such that~$\delta \simplex=0$ is added to~$\ClassCur$, then it is a birth simplex and creates a new pair~$[\simplex,\infty]$;
    \item If a simplex~$\simplex$ such that~$\low (\delta \simplex)=\simplex'\in \FiltrationIn\cup \ClassCur$ is added to~$\ClassCur$, then it is a death simplex, which has the effect to replace a pair~$[\simplex',\infty]$ in~$\pdata$ with a new pair~$[\simplex',\simplex]$.
\end{itemize}

Given a complete~$D$-compatible classification, recall that we have a correspondence~$\Pairing$ of intervals in~$D$ with critical pairs. When~$\pdata$ is a (non-complete) partial $D$-compatible classification, the correspondence is not necessarily defined for all intervals in~$D$: For instance there are no birth and death critical simplices for~$(b,d)$ in~$\pdata$, for any interval~$(b,d)\in D$ such that~$j<b<d$. In this case we only have a partially defined correspondence, which we indicate by the conventions~$\PairingBirth(b,d):=\emptyset$ and~$\PairingDeath(b,d):=\emptyset$. Note that it is also possible to have~$\Pairing(b,d)=(\simplex,\emptyset)$ whenever we already have a birth critical simplex~$\simplex=\PairingBirth(b,d)$ in~$\pdata$ that is not yet associated with a death critical simplex. In the algorithm, we store and update the partially defined correspondence~$\Pairing$ as we incrementally construct the partial classification~$\pdata$.

Finally, to improve the time-efficiency of our algorithm, we maintain an array~$\MaxFree \in \N^{\dimK+1}$ of integers that constrain the non-critical simplices that can be added to the classification in each dimension. In the array~$\MaxFree=(\MaxFree_0,\cdots, \MaxFree_{\dimK})$, each~$\MaxFree_{p}$ is the number of non-critical positive $p$-simplices that remain to be added in the classification. Since non-critical simplices come in pairs,~$\MaxFree_{p}$ is also the number of non-critical negative $(p+1)$-simplices that remain to be added in the classification: At initialization this is the rank~$\mathrm{rank}(\partial_{p+1})$ of the boundary matrix restricted to~$(p+1)$-simplices, minus the number of bounded bars~$(b,d)\in D$ in degree~$p$ (because those bars are in 1-1 correspondence with negative critical~$(p+1)$-simplices). By convention~$\MaxFree_{-1}=\MaxFree_{\dimK+1}=0$.

In practice we call the main algorithm $\AlgoComputeFiltrations(\SComplex,D)$ (see Alg~\ref{alg:compute_filtrations}) to build the initial partial classification and then call the exploration (see Alg.~\ref{alg:DFS_algo_compatible_data}) as a subroutine $\DFS(\pdata,\Results,\delta,\Pairing,\MaxFree)$.
\begin{algorithm}[ht]
 \caption{ $\AlgoComputeFiltrations(\SComplex,D)$ }
 \label{alg:compute_filtrations}
\begin{algorithmic}[1]
\REQUIRE Finite simplicial complex~$\SComplex$ and target barcode~$D$
\ENSURE List~$\Results$ of all $D$-compatible classifications
\STATE $\pdata= \big( \FiltrationIn,j, \ClassCur,  \IntervalsCur,\TrivCur)\gets (\emptyset,\emptyset,1,\emptyset, D^{<2},\emptyset)$
\STATE{Initialise~$\delta$ as the boundary matrix of~$\SComplex$}
\STATE{Initialise correspondence~$\Pairing(b,d):=[\emptyset,\emptyset]$ for bounded~$(b,d)\in D$ and~$\Pairing(b,\infty):=[\emptyset,\infty]$ for infinite~$(b,\infty)\in D$}%
\STATE $\MaxFree_{p}\gets \mathrm{rank}(\partial_{p+1})- \sharp \big\{(b,d)\in D \, | \, \dim (b,d)= p+1, \, d<\infty \big\}$, for $0\leq p \leq \dimK$ 
\RETURN $\DFS(\pdata,\Results,\delta,\Pairing,\MaxFree)$
\end{algorithmic}
\end{algorithm}

\begin{algorithm}[ht]
 \caption{ $\DFS(\pdata, \Results,\delta,\Pairing,\MaxFree)$ }
 \label{alg:DFS_algo_compatible_data}
\begin{algorithmic}[1]
\REQUIRE Partial classification $\pdata$, list~$\Results$ of $D$-compatible classifications, matrix~$\delta$, pairing~$\Pairing$, integers~$\MaxFree\in \N^{d+1}$
\ENSURE List~$\Results$ enriched with $D$-compatible classifications that extend the classification~$\FiltrationIn$
\IF{$|\FiltrationIn|=\sharp \SComplex$ and $\Din=D$ and $\ClassCur=\emptyset$} \RETURN $\Results\cup \{\FiltrationIn{}\}$ 
\ELSIF{$\ClassCur\neq\emptyset$ and $\IntervalsCur= \emptyset$ and $\TrivCur=\emptyset$}
\STATE{$\FiltrationIn'\gets \FiltrationIn\cup \ClassCur$}
\STATE{{\bf if} $\ClassCur$ contains critical simplices for the endpoint~$j$ {\bf then} $j'\gets j+1$ {\bf else} $j'\gets j$}
\IF{$\Dinbis \neq D$} \STATE $\IntervalsCur'\gets \{(b,d)\in D | \, b=j' \text{ or } d=j'\}$
\STATE $\pdata' \gets (\FiltrationIn',j', \emptyset,\emptyset,\IntervalsCur')$
\STATE $\Results\gets \DFS(\pdata',\Results,\delta,\Pairing,\MaxFree)$ \ENDIF
\STATE $\pdata' \gets (\FiltrationIn',j', \emptyset,\emptyset,\emptyset)$ 
\RETURN $\DFS(\pdata',\Results,\delta,\Pairing,\MaxFree)$
\ELSE
\FOR{interval~$(b,d)$ in~$\IntervalsCur$}
\IF{$b=j$}
\FOR{simplex~$\simplex \notin \FiltrationIn \cup \ClassCur$ such that~$\delta\simplex=0$ and~$\dim \simplex= \dim (b,d)$}
\STATE $\ClassCur'\gets \ClassCur \cup \{\simplex\}$ and $\IntervalsCur'\gets \IntervalsCur\setminus \{(b,d)\}$
\STATE $\pdata' \gets (\FiltrationIn, j, \ClassCur', \IntervalsCur', \TrivCur)$
\STATE $\Pairing'\gets \Pairing$ where $\Pairing'$ agrees with~$\Pairing$ on all inputs, except that  $\PairingBirth'(b,d):=\simplex$ instead of $\emptyset$
\STATE $\Results\gets \DFS(\pdata',\Results,\delta,\Pairing',\MaxFree)$
\ENDFOR
\ELSIF{$d={j}$}
%
\FOR{simplex~$\simplex\notin \FiltrationIn \cup \ClassCur$ such that~$\low (\delta\simplex)=\PairingBirth(b,d)$ and~$\dim \simplex= \dim (b,d) +1$}
\STATE $\ClassCur'\gets \ClassCur \cup \{\simplex\}$ and $\IntervalsCur'\gets \IntervalsCur\setminus \{(b,d)\}$
\STATE $\pdata' \gets (\FiltrationIn, j, \ClassCur', \IntervalsCur', \TrivCur)$
\STATE $\delta' \gets$ reduce~$\delta$ with respect to the column of $\simplex$, c.f.\ Remark \ref{remark_matrix_reduction}.

\STATE $\Pairing'\gets \Pairing$, where $\Pairing'$ agrees with $\Pairing$ on all inputs, except that  $\PairingDeath'(b,d):=\simplex'$ instead of $\emptyset$
\STATE $\Results\gets \DFS(\pdata',\Results,\delta',\Pairing',\MaxFree)$ 
\ENDFOR
\ENDIF
\ENDFOR

\FOR{simplex~$\tau \in \TrivCur$}
\FOR{simplex~$\simplex\notin \FiltrationIn \cup \ClassCur$ such that~$\low (\delta\simplex)=\tau$ and~$\dim \simplex= \dim \tau +1$}
\STATE $\pdata' \gets (\FiltrationIn,j , \ClassCur\cup \{\simplex\},\IntervalsCur, \TrivCur\setminus\{\tau\})$
\STATE $\delta' \gets$ reduce~$\delta$ with respect to the column of~$\simplex$, c.f. Remark \ref{remark_matrix_reduction}
\STATE $\Results\gets \DFS(\pdata',\Results,\delta',\Pairing,\MaxFree)$ \ENDFOR
\ENDFOR
\FOR{simplex~$\simplex\notin \FiltrationIn \cup \ClassCur$ such that~$\delta\simplex=0$}

\IF{$\MaxFree_{\dim \simplex}>0$}
\STATE $\MaxFree_{\dim \simplex }\gets \MaxFree_{\dim \simplex} -1$
\STATE $\pdata' \gets (\FiltrationIn,j , \ClassCur\cup \{\simplex\}, \IntervalsCur, \TrivCur\cup \{\simplex\})$
\STATE $\Results\gets \DFS(\pdata',\Results,\delta,\Pairing,\MaxFree)$
\ENDIF \ENDFOR
\RETURN $\Results$ \ENDIF

\end{algorithmic}
\end{algorithm}
\begin{theorem}
\label{theorem_algo_correct}
Algorithm~$\AlgoComputeFiltrations(\SComplex,D)$ returns the list of all~$D$-compatible classifications.
\end{theorem}
\begin{proof}
From Line~$1$ of Alg.~\ref{alg:DFS_algo_compatible_data}, the output consists in the $D$-compatible classifications~$\FiltrationIn$ extracted from complete partial classifications~$\pdata$ that are encountered by the algorithm. Thus it suffices to show that any given partial classification~$\pdata$ is explored by the algorithm. 
We proceed by induction on the number~$m$ of classes~$\Class_1,\cdots,\Class_m$ forming~$\FiltrationIn$, and on the number of simplices in~$\Class_1 \sqcup \cdots \sqcup \Class_m \sqcup \ClassCur$. Note that the empty partial classification is visited at the beginning of the algorithm, so we may assume that~$\Class_1 \sqcup \cdots \sqcup \Class_m \sqcup \ClassCur\neq \emptyset$. 
If~$\ClassCur=\emptyset$, we have~$m\geq 1$ and~$\Phi_m\neq \emptyset$. We can then form the partial classification~$\pdata'$ with~$\FiltrationIn':=\Class_1\sqcup \cdots \sqcup \Class_{m-1}$,~$\ClassCur':=\Class_m$,~$\IntervalsCur':=\emptyset$,~$(\TrivCur)':=\emptyset$, and~$j':=j-1$ or~$j':=j$  depending on whether~$ \Class_m=\Class_{i_{j-1}}$ or not. Clearly then, by Lines~$3-10$ of Alg.~\ref{alg:DFS_algo_compatible_data}, if~$\pdata'$ is explored by the algorithm, then so is~$\pdata$. 
Otherwise,~$\ClassCur\neq \emptyset$ has a last simplex~$\simplex$. We can then form~$\pdata'$ by removing~$\simplex$, i.e.~$\ClassCur':=\ClassCur \setminus \{\simplex\}$, with~$\FiltrationIn':=\FiltrationIn$ and~$j'=j$ unchanged, and with the obvious changes in~$\IntervalsCur'$ if~$\simplex$ were a critical simplex, or in $(\TrivCur)'$ if it were non-critical. Then~$\pdata$ is one of the four incremental extensions of~$\pdata'$ depicted in Fig.~\ref{fig:extensions}, therefore by Lines~$14-37$ of Alg.~\ref{alg:DFS_algo_compatible_data}, if~$\pdata'$ is explored by the algorithm, then so is~$\pdata$. 
\end{proof}

\begin{remark}
\label{remark_BFS_vs_DFS}
The exploration of Algorithm~\ref{alg:DFS_algo_compatible_data} is equivalently viewed as a Depth-First Search (DFS) on {\em the tree of partial~$D$-compatible classifications}: each node is a partial classification, whose children differ by the addition of a single simplex according to the four types of extension depicted in Fig.~\ref{fig:extensions}. The algorithm records the subset of leaves that correspond to~$D$-compatible classifications. It would also have been possible to design a Breadth-First-Search (BFS) algorithm. However the BFS approach requires more storage, because we can forget the information stored in a node (e.g. the boundary matrix of a partial classification) only when all its children are treated. Hence in a BFS version of the algorithm we would eventually need to store the information of the entire tree, while in the DFS version at most one branch is stored at a time.
\end{remark}
\begin{remark}
\label{remark_matrix_reduction}

We implicitly maintain a matrix factorization~$\delta = \partial V$ at each step of the algorithm, where $\partial \in \field^{\SComplex \times \SComplex}$ is the total boundary matrix of~$\SComplex$ and $\delta, V \in \field^{\SComplex \times \SComplex}$ are square.  This factorization must satisfy two conditions, each of which is stated in terms of a sequence of form $\xi = (\simplex_1, \ldots, \simplex_p, \tau_1, \ldots, \tau_q, \upsilon_1, \ldots, \upsilon_r)$,  where $\simplex_1 < \cdots < \simplex_p$ is the linear order on $\FiltrationIn $ and $\upsilon_1 < \ldots < \upsilon_q$ is the linear order on $\ClassCur$.  We write $\hat \partial$ for the matrix obtained  by permuting rows and columns of $\partial$ such that simplex~$\xi_k$ indexes the $k$th row and column of~$\hat \partial$, for each~$k$.  Matrices~$\hat \delta$, and $\hat V$ are defined similarly, by permuting rows and columns to ensure that~$\xi_k$ indexes the $k$th row and column of each matrix.  Our two conditions can now be stated as follows:
    \begin{itemize}
        \item Matrix $\hat V$ must be upper unitriangular.
        \item Matrix  $\hat \delta$ must be \emph{partially reduced} in the following sense.  For each birth-death pair~$[\tau ,\tau']$ in~$\FiltrationIn \cup \ClassCur$ (including non-critical pairs and excluding pairs of form $[\tau, \infty]$), the entry $\hat \delta (\tau, \tau')$ must be nonzero, and each entry that lies either directly below~$\hat \delta (\tau, \tau')$ in column~$\tau'$ (respectively, each entry to the right of~$\hat \delta(\tau, \tau')$ in row~$\tau$) must equal zero.
    \end{itemize}
It can be shown that the low function of $\hat \delta$ agrees with the low function of any~$R = DV$ decomposition  of$\hat \partial$ (when restricting each of these functions to the subset $\FiltrationIn \cup \ClassCur$; their values may differ for $\tau \notin \FiltrationIn \cup \ClassCur$), c.f. \cite{cohen2006vines}.  

To obtain such a factorization after we have added $\simplex$ to $\ClassCur$,  first fix a compatible sequence $\xi$ and permute the columns of $\delta$ accordingly;  then perform one further swap to ensure that the column indexed by $\simplex$ appears directly to the right of the last  column indexed by $\FiltrationIn$, keeping the location of all columns indexed by $\FiltrationIn$ fixed.  Perform the same permutation on rows.  Then add multiples of column $\simplex$ to columns on its right as necessary to ensure that the unique nonzero entry in row $\low(\delta \simplex)$ appears in column $\simplex$.  A routine exercise shows that the resulting matrix $\hat \delta'$, fits into a matrix factorization $\hat \delta' = \hat \partial  \hat V$ of the appropriate form. 
\end{remark}
\begin{remark}
\label{remark_nerve}
From Theorem~\ref{theorem_description_polytopes} the polyhedra~$\Polytope(\Filtration)$ induced by the~$D$-compatible classifications~$\Filtration$ describe the fiber~$\persmap^{-1}(D)$ as a polyhedral complex. In applications where it is desirable to dispose of a simplicial complex structure for~$\persmap^{-1}(D)$, we can simply form the nerve of the cover associated to the polyhedra~$\Polytope(\Filtration)$, which by the Nerve theorem for Euclidean closed convex sets is homotopy equivalent to~$\persmap^{-1}(D)$ (see e.g.~\cite[\S 5]{weil1952theoremes}). In practice, computing the nerve amounts to finding non-empty intersections~$\bigcap_{l=1}^n \Polytope(\Filtration_l)$, which in a polyhedral complex boils down to intersecting vertex sets:
\[\VertSet \left (\bigcap_{l=1}^n \Polytope(\Filtration_l) \right)= \bigcap_{l=1}^n \VertSet(\Polytope(\Filtration_l)).\]
We include the construction of this simplicial complex for describing~$\persmap^{-1}(D)$ in our implementation.
\end{remark}
\begin{remark}
Since the polyhedral decomposition of the fiber realizes a regular CW complex, computing the~$\Z_2$ linear boundary operator of this object reduces to enumerating the codimension-1 faces of each polyhedron. These may be computed from the standard formula for the boundary of a product of copies of standard geometric simplices:
    \begin{align*}
        \textstyle
        \partial(\sigma_1 \times \cdots \sigma_m)
        =
        \bigcup_k [\sigma_1 \times \cdots \times \partial(\sigma_k) \times \cdots \times \sigma_m]
    \end{align*}
Computing the coefficients of the $\Z$-linear boundary matrix is slightly more involved, due to orientation.  We deffer this problem to later work.
\end{remark}

\begin{remark}[Generalization to barcodes for persistent (relative) (co)homology]  The discussion this far has focused exclusively on barcodes of the homological persistence module (obtained by applying the homology functor to a nested sequence of cell complexes). However, there are several other formulae for generating persistence modules from a filtered cell complex.  While each construction has distinct and useful algebraic properties, their barcodes are completely determined by 
that of the homological barcode \cite{de2011dualities}. Thus the procedure to compute fiber of~$\persmap$ also serves to compute the persistence fiber of these other constructions. For a detailed discussion, see Appendix~\ref{sec:dualities}.
\end{remark}

\section{Generalisation to based chain complexes}
\label{section_based_chain_complexes}

We generalise the fact that the fiber of the persistence map~$\persmap$ is a polyhedral complex to filters defined directly at the level of based chain complexes. These include filters on simplicial complexes, cubical complexes, delta complexes and CW complexes. In turn our approach for computing~$\persmap^{-1}(D)$ adapts to these situations as well.
\begin{definition}
A based, finite-dimensional, $\field$-linear chain complex is a pair $(\chaincx, \basis)$ such that $\sum_i \dim(\chaincx_i) < \infty$ and $E$ is a union of bases $E_i$ of $C_i$ for all $i$.
A \emph{filter} on~$(\chaincx, \basis)$ is a real-valued function $\filta: \basis \to \I $ such that the linear span of $\{ e \in \basis : \filta(e) \le t\}$ forms a linear subcomplex of $\chaincx$, for each $t \in \I$.
\end{definition}

Here are some examples of~$(\chaincx, \basis)$ induced by combinatorial complexes:
\begin{enumerate}
    \item \textbf{Simplicial Complexes.} Basis~$\basis$ is the collection of simplices in a simplicial complex~$\SComplex$. We recover the standard setting of filters over~$\SComplex$.
    \item \textbf{Cubical complexes.} Basis~$\basis$ is the collection of cubes in a cubical complex.
    \item \textbf{Delta and CW Complexes.} Basis~$\basis$ is the collection of cells in a delta complex or CW complex~$\SComplex$.
\end{enumerate}
These variations are of interest in practice: For instance with delta and CW complexes we can decompose topological spaces with much fewer simplices, while cubical complexes appear naturally e.g. in image analysis. The main result of this section, Theorem \ref{thm_CW_complexes_have_good_fibers}, generalizes the structure theorem for simplicial complexes (Theorem \ref{theorem_polyhedral_complex}) to these important variants.  In particular, Theorem \ref{thm_CW_complexes_have_good_fibers} implies each of the following results.
\begin{theorem}
\label{thm:generalization_to_cw_complexes}
Let~$\SComplex$ be a simplicial, cubical, delta or CW complex and let~$D$ be a barcode. Then the fiber~$\persmap^{-1}(D)$ is the underlying space of a polyhedral complex whose polyhedra are products of standard simplices.
\end{theorem}

\begin{theorem}
Theorem \ref{thm:generalization_to_cw_complexes} remains true if we restrict to lower-star filtrations or Vietoris-Rips filtrations.
\end{theorem}

\subsection{Polyhedral decomposition of the ambient cube}
Let~$\basis$ be a finite set, and let~$\critvals = \{\critval_0 < \cdots < \critval_\ncritval \}$ be a finite subset of the  interval~$\I$, where $\critval_0 = 1$ and $\critval_m \leq \sharp \SComplex$. For any pair of functions~$\filta, \filtb: \basis \to \I$, let us define the relation~$\filta \sim_\critvals \filtb$ by either of the following two equivalent axioms:
    \begin{enumerate}[label=\textbf{(A\arabic*)}]
        \item 
        \label{item:cellaxiom_compose}        
        There exists an order-preserving map $\psi: \I \to \I$ such that $\psi(\critval_i) = \critval_i$ for each $i$, and $\filtb = \psi \circ \filta$.    
        \item 
            \label{item:cellaxiom_inequality}
            The function $\filtb$ satisfies
            \begin{align}
                \filta(e) \in \critvals \implies \filtb(e) = \filta(e)
                &&
                \filta(e) \le \filta(e') \implies \filtb(e) \le \filtb(e').
                \label{eq:equivalence_condition}
            \end{align}        
    \end{enumerate}
Note that the relation $\sim_\critvals$ is reflexive and transitive but not symmetric.
We write
\[
\cell_\critvals(\filta) := \{ \filtb \in \I ^E  : \, \, \filta \sim_\critvals \filtb \}.
\]
\begin{lemma}
\label{lem:cellispoly}
The set  $\cell_\critvals(\filta)$ is a compact polyhedron.  
\end{lemma}
\begin{proof}
Axiom \ref{item:cellaxiom_inequality} represents a family of logical conditions, each of which determines either a hyperplane, i.e. $\{\filtb : \filtb(e) = c\}$, or a closed half-space, i.e. $\{ \filtb : \filtb(e') - \filtb(e) \ge 0\}$.  The intersection of these sets, $\cell_\critvals(\filta)$, is a bounded polyhedron.
\end{proof}

\begin{lemma}
\label{lem:intersectionimpliescontainment}
If $\filtb \in \cell_\critvals(\filta)$ then $ \cell_\critvals(\filtb) \subseteq \cell_\critvals(\filta)$.
\end{lemma}
\begin{proof}
Relation $\sim_\critvals$ is transitive.
\end{proof}

\begin{proposition}
\label{cor:polyhedral_decomp_unit_square}
The set of convex polyhedra $\ambientPoly_\critvals = \{ \cell_\critvals(\filta) : \filta: \basis \to \I \}$ is a polyhedral complex; the underlying space is $\vert \ambientPoly_\critvals \vert  = \I^\basis$.  
\end{proposition}
\begin{proof}
Follows from Lemmas \ref{lem:cellispoly} and  \ref{lem:intersectionimpliescontainment}.
\end{proof}

Next we provide the description of each polyhedron~$\cell_\critvals(\filta)$ in the complex as a product of standard simplices. For convenience let $\critval_{m+1}:=\infty$, and for each~$i \in \{1, \ldots, m\}$, let~$\eta^i = \filta(\basis) \cap (\critval_i, \critval_{i+1})$.

\begin{lemma}
\label{lem:polyhedraareproductsofsimplices}
If each $\eta^i$ is nonempty, then $\cell_\critvals(\filta)$ is affinely isomorphic to $\Delta_{\# \eta^0 } \times \cdots \times \Delta_{\# \eta^m}$, where~$\Delta_{k}$ stands for the standard geometric simplex of dimension~$k$.
\end{lemma}
\begin{proof}
Follows from Axiom \ref{item:cellaxiom_compose}.
\end{proof}

\subsection{Polyhedral decomposition of $\persmap^{-1}(D)$ for based chain complexes}

Let $(\chaincx, \basis)$ be a based, finite-dimensional, $\field$-linear chain complex.

\begin{lemma}
\label{lem:wholecell}
Let $\filta$ be a filter on $(\chaincx, \basis)$; $D$ be the associated (total) barcode; and~$\critvals:=\FinEnd(D)$ be the set of finite endpoints of intervals in $D$.  Then each element of~$\cell_\critvals(\filta)$ is a bona fide filter on $(\chaincx, \basis)$ with barcode $D$.
\end{lemma}
\begin{proof}
The persistence map is equivariant in the sense that~$\persmap(\psi \circ f)= \psi.\persmap(f)$ for~$\psi: \I\rightarrow \I$ a non-decreasing map, as can be proven e.g. following~\cite[Lemma~1.5]{leygonie2021fiber}. Here~$\psi$ acts point-wise on intervals of~$\persmap(f)$, i.e. each interval~$(b,d)\in \persmap(f)$ produces an interval~$(\psi(b),\psi(d))$ in~$\psi.\persmap(f)$ whenever~$\psi(b)\neq \psi(d)$. The result follows from Axiom~\ref{item:cellaxiom_compose}. \qedhere
\end{proof}
Given~$\unionofpolytopes$ a union of polyhedra in a polyhedral complex~$\ambientPoly$, then the subcomplex induced by~$\unionofpolytopes$ is defined as
    \begin{align}
        \ambientPoly[S] = \{ X \in \ambientPoly : X \subseteq \unionofpolytopes \}.      
    \end{align}
\begin{theorem}
\label{thm:fiberispoly}
Let $(\chaincx, \basis)$ be a point-wise finite dimensional based chain complex;~$D$ be a barcode;~$\critvals:=\FinEnd(D)$; and~$\filtersCompat$ be the set of filters on~$(\chaincx, \basis)$ with barcode~$D$. Then~$\filtersCompat$ is a union of polyhedra in~$\ambientPoly_{\critvals}$, hence there exists a well-defined polyhedral subcomplex
    \begin{align*}
        \ambientPoly_\critvals[\filtersCompat] \subseteq \ambientPoly_\critvals 
    \end{align*}
with underlying space $\vert \ambientPoly_\critvals[\filtersCompat] \vert = \filtersCompat $.
\end{theorem}
\begin{proof}
Lemma \ref{lem:wholecell} implies that~$\filtersCompat$ is a union of the polyhedra in $\{ \cell_\critvals(\filta) : \filta \in \filtersCompat \}$, and by continuity it is a closed subset of the polyhedral complex~$\ambientPoly_\critvals=\{ \cell_\critvals(\filta) : \filta: \basis \to \I \}$ (Proposition~\ref{cor:polyhedral_decomp_unit_square}). Therefore~$\filtersCompat$ is a sub-polyhedral complex.  
\end{proof}

\subsection{Polyhedral decomposition of $\persmap^{-1}(D)$ for CW complexes}

The polyhedral decomposition of~$\persmap^{-1}(D)$ for based chain complexes (Theorem~\ref{thm:fiberispoly}) does not carry over directly to arbitrary CW complexes because given a CW complex~$\cwcomplex$, there may exist filters on the associated based chain complex~$(\chaincx, \basis)$ that do not correspond to valid filters of~$\cwcomplex$. Here the notion of filter on a CW complex (respectively, delta or cubical complex) naturally generalises the simplicial situation: it qualifies any function~$\filta: \cwcomplex\rightarrow \I$ whose sub-level sets are sub-CW complexes (respectively, sub-delta or sub-cubical complexes).

\begin{example}
Let $\cwcomplex = \{ v, e\}$ be the CW decomposition of $S^1$  with one vertex, $v$,  and one edge, $e$. Since all boundary maps are 0, the filtration $\{e\} \subseteq \{v, e\}$ is perfectly valid for $(\chaincx, \basis)$, but not for $\cwcomplex$.
\end{example}

Fortunately this problem is simple to address.  Lemma \ref{lem:cellinheritance} represents the only technical observation needed to extend our polyhedral characterization of the persistence fiber from regular finite CW complexes to arbitrary finite CW complexes.  The proof is vacuous.

\begin{lemma}
\label{lem:cellinheritance}
Suppose that a function $\prop: \I^\basis \to \{\True, \; \False \}$ satisfies the condition that 
    \begin{align}
        \prop(\filta) = \True \implies \prop \text{ evaluates to } \True \text{ on each element of  }\cell_\critvals(\filta).
        \label{eq:allifone}
    \end{align}

Then $\prop^{-1}(\True)$ is a union of polyhedra in $\ambientPoly_\critvals$, hence $\ambientPoly_\critvals[  \prop^{-1}(\True)] $ is a polyhedral subcomplex.
\end{lemma}

\begin{theorem}
Let~$\SComplex$ be a simplicial, cubical, delta or CW complex.  Let~$D$ be a barcode, and let~$(\chaincx, \basis)$ be the induced based chain complex. Then the fiber~$\persmap^{-1}(D)$ is a polyhedral complex whose polyhedra are products of standard simplices.
\end{theorem}
\begin{proof}
Let $\cwcomplex$ be a finite CW complex, and let $\prop(\filta) = \True$ iff 
$\filta^{-1}(-\infty, t]$ is a CW subcomplex of $\cwcomplex$ for each $t \in \I$.  Then $\prop(\filta) \implies \prop(\filtb)$ for each $\filtb$ such that $\filta \sim_\critvals \filtb$, by Axiom \ref{item:cellaxiom_compose}.   Lemma \ref{lem:cellinheritance} therefore implies that $\persmap^{-1}(D)$ is the underlying space of a polyhedral complex.  The polyhedra in this complex are isomorphic to products of simplices, by Lemma \ref{lem:polyhedraareproductsofsimplices}. 
\end{proof}

In practice, it is simple to adapt the $\DFS$ algorithm to compute $\persmap^{-1}(D)$ for a CW complex.  One must simply add the condition that $\FiltrationIn \cup \{ \simplex \}$ be a bona fide subcomplex to each of the lists of criteria found on lines 15, 21, 28, and 32.
\subsection{Polyhedral decomposition of $\persmap^{-1}(D) \cap S$, for a restricted family $S$}
\label{sec_technical_results_for_restricted_filts}

It is common in practice to restrict one's attention to a restricted family of filters, $S$.  A noteworthy family of examples come from the \emph{lower-$p$ filtrations}.  Given a CW complex $\cwcomplex$ with associated based chain complex $(\chaincx, \basis)$, we define the space of \emph{lower-$p$ filtrations} on $\cwcomplex$ as 
    \begin{align*}
        \filterLowerK{p}{\cwcomplex}
        =
        \{
        \filta: \basis \to \I : \text{ for each cell } \tau \in \basis, \text{ one has }\filta(\tau) = \max \{\filta(\sigma) : \sigma \subseteq \mathrm{cl}(\tau), \; \dim(\sigma) \le p \} 
        \}
    \end{align*} 
Examples include

\begin{enumerate}
    \item The space of lower-star filtration,  $\filterLowerK{0}{\cwcomplex}$.  
    \item  The space of lower-edge filtrations,  $\filterLowerK{1}{\cwcomplex}$.  These include the so-called Vietoris-Rips filtrations of a combinatorial simplicial complex.
\end{enumerate}

The important fact about these spaces is that they are unions of polyhedra in the polyhedral decomposition of the ambient cube:

\begin{lemma}
\label{lem:lower_star_is_union_of_polys}
For each $p$ and each finite subset $\critvals \subset \I$, the space $\filterLowerK{p}{\cwcomplex}$ is a union of polyhedra in $\ambientPoly_\critvals$.
\end{lemma}
\begin{proof}
Fix $\filta \in \filterLowerK{p}{\cwcomplex}$, and suppose that $\filta \sim_\critvals \filtb$ for some $\filtb$.  Then it follows from  Axiom \ref{item:cellaxiom_compose} that $\filtb \in \filterLowerK{p}{\cwcomplex}$. In particular, if $\filterLowerK{p}{\cwcomplex}$ contains an element of a polyhedron in $\ambientPoly_\critvals$, then it contains the entire polyhedron.
\end{proof}

Then we have the following.
\begin{theorem}
\label{thm_CW_complexes_have_good_fibers}
Suppose that $(\chaincx, \basis)$ is the based chain complex associated with a finite simplicial complex, a delta complex or a CW complex~$\cwcomplex$. Let $D$ be a barcode and  $\critvals = \FinEnd(D)$ be the associated set of  finite endpoints. Then
    \begin{align*}
        S = \filtersCompat \cap \filterLowerK{p}{\cwcomplex}
    \end{align*}
is a union of polyhedra in $\ambientPoly_\critvals$.  Consequently there is a nested sequence of polyhedral subcomplexes
    \begin{align}
        \label{eq:lower_k_fiber}
        \ambientPoly_\critvals[S]
        \subseteq
        \ambientPoly_\critvals[\filtersCompat]
        \subseteq 
        \ambientPoly_\critvals. 
    \end{align}
When $p = 0$, the set $S$ is the space of lower-star filters with barcode~$D$.  When $p = 1$, the set $S$ is the space of lower-edge filters with barcode~$D$.
\end{theorem}
\begin{proof}
Follows from Lemma \ref{lem:lower_star_is_union_of_polys}.
\end{proof}

\begin{remark}[Adapting the $\DFS$ algorithm]
\label{remark_lower_star_fiber}
If we wish only to compute a polyhedral decomposition of $D$-compatible lower-$p$ filters, then the exploration of Alg.~\ref{alg:DFS_algo_compatible_data} can be considerably reduced. For this we maintain a list~$\LowerStarSimplices$ of simplices~$\simplex\notin \FiltrationIn$ whose faces~$\tau$ with $\dim(\tau) \leq p$ satisfy~$\tau\in \FiltrationIn$;  we deem the current class~$\ClassCur$ incomplete as long as~$\LowerStarSimplices$ is non-empty. In practice this simply amounts to modifying Line 3 to:
\[\text{(Line 3') \hspace{.4cm}  {\bf else if }} \LowerStarSimplices=\emptyset \text{ and } \ClassCur\neq \emptyset  \text{ and } \IntervalsCur= \emptyset  \text{ and } \TrivCur= \emptyset \text{ {\bf then}}\]
In addition we need to update the list~$\LowerStarSimplices$ each time we add a new vertex~$\tau$ in the classification (which can happen at lines~16,~22,~29 and~35) simply by adding the following line:
\[\text{Add to list } \LowerStarSimplices \text{ all simplices } \simplex \text{ such that } \FiltrationIn\cup \ClassCur \cup \{\tau\} \text{ contains the $k$-dimensional faces of~$\simplex$ for each $k \leq p$}.\]
\end{remark}
\section{Experiments}
\label{section_experiments}
Using our algorithm we compute the number of polyhedra in~$\persmap^{-1}(D)$, binned by dimension, and the non-zero Betti numbers~$\beta_p(\persmap^{-1}(D))$. We also record the special cases where the facets, i.e. polyhedra in~$\persmap^{-1}(D)$ that are maximal w.r.t. inclusions, do not consist only in the top-dimensional polyhedra. 
\paragraph{Implementation} Algorithm \ref{alg:compute_filtrations} was implemented in the programming language Rust.  This implementation accommodates user-defined coefficient fields, based complexes, and restricted families of filters, e.g. lower-star.  The implementation uses several dependencies from the ExHACT library \cite{ExHACT} for low-level functions, including reduction of boundary matrices and implementation of common coefficient fields. Source code for both libraries is to be made available in the near future.  
\paragraph{Reading the results} The outputs of the algorithm are reported in figures. Each figure corresponds to a specific simplicial, or CW complex and provides statistics about the fiber~$\persmap^{-1}(D)$ for various barcodes in a table. By convention, black intervals in the target barcode~$D$ are of dimension~$0$, blue intervals are of dimension~$1$, while green intervals are of dimension~$2$. In all cases the number of polyhedra in~$\persmap^{-1}(D)$ is binned by dimension in the form of an array, and the Betti numbers are computed with coefficients in~$\mathbb Z_2$. Unless explicitly stated otherwise:
\begin{itemize}
    \item The facets are the top dimensional polyhedra. Otherwise, we explicitly report in red the facets binned by dimension in an array;
    \item The persistence modules associated to barcodes~$D$ are computed with coefficients in~$\mathbb Z_2$. Otherwise, in special cases where we also compute persistence modules with coefficients in~$\mathbb Q$, the coefficient field is indicated in blue by a specific mention;
    \item The fiber is computed inside the space of all filters. Otherwise, we provide as many columns for statistics about~$\persmap^{-1}(D)$ as there are categories of filters to consider. 
\end{itemize}
\subsection{Simplicial Complexes}
\label{subsection_experiments_simplicial_complexes}
In all the examples of this section~$\SComplex$ is a simplicial complex. When~$\SComplex$ is a tree (Figure~\ref{fig:fiber_trees}), we report these statistics both when the domain of~$\persmap$ consists of all filters and when it is restricted to lower star filters. For lower star filters on the interval the fiber is shown by~\cite{cyranka2018contractibility} to consist of contractible components. Our computations indicate that this property holds as well for general filters on the interval, however it breaks for other trees where the fiber has loops as indicated by non-trivial Betti numbers~$\beta_1$. 
\begin{figure}[ht]
\centering
\includegraphics[width=0.8\textwidth]{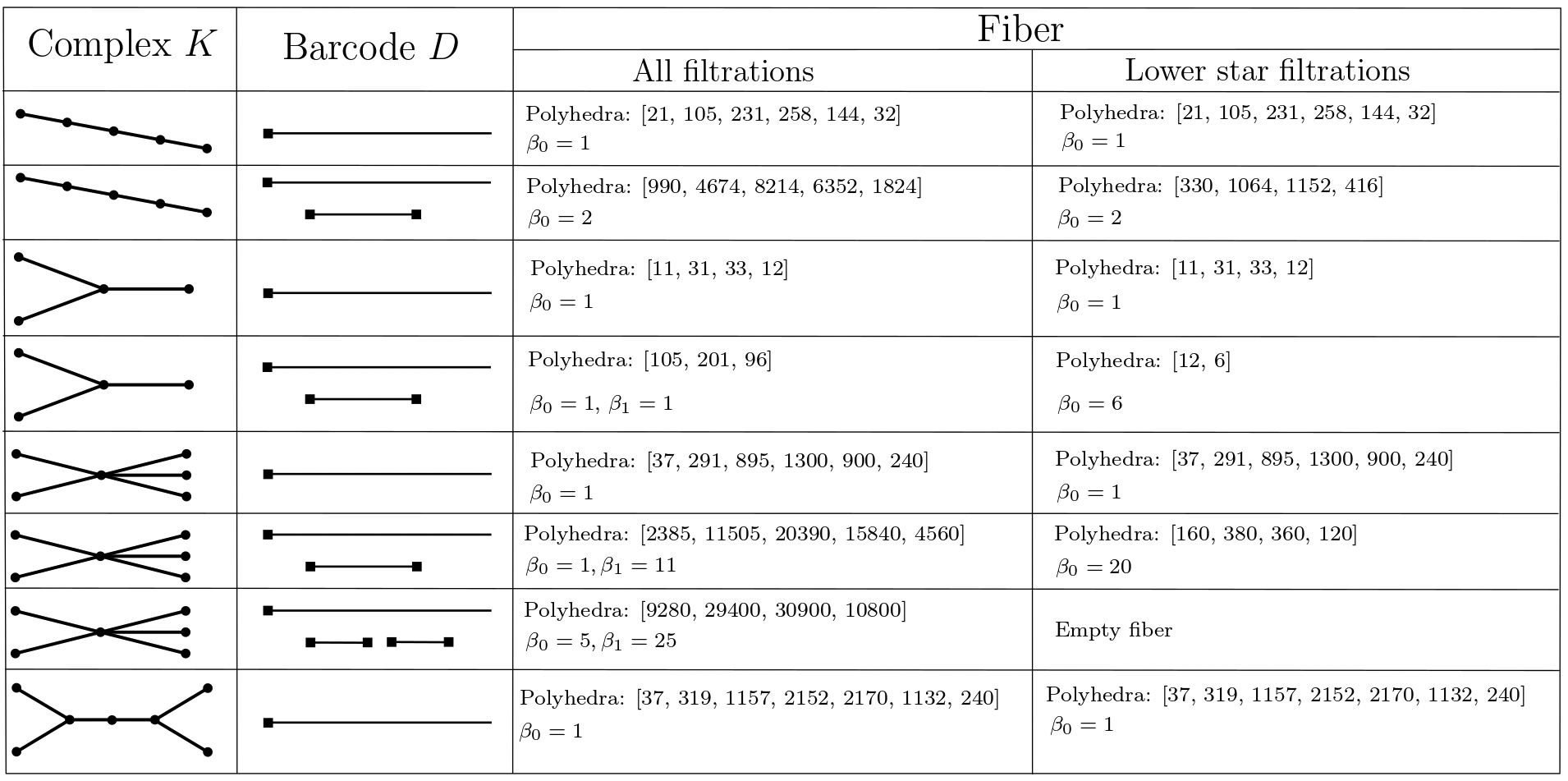}
\caption{Some statistics about fibers~$\persmap^{-1}(D)$ when~$\SComplex$ is a tree. }
\label{fig:fiber_trees}
\end{figure}

For lower star filters on arbitrary subdivisions of the circle it is proven in~\cite{mischaikow2021persistent} that the fiber is made of circular components. Our computations (Figure~\ref{fig:fiber_triangle_and_exts}) suggest that this property holds as well when allowing general filters and adding dangling edges to the circle. 
\begin{figure}[ht]
\centering
\includegraphics[width=0.55\textwidth]{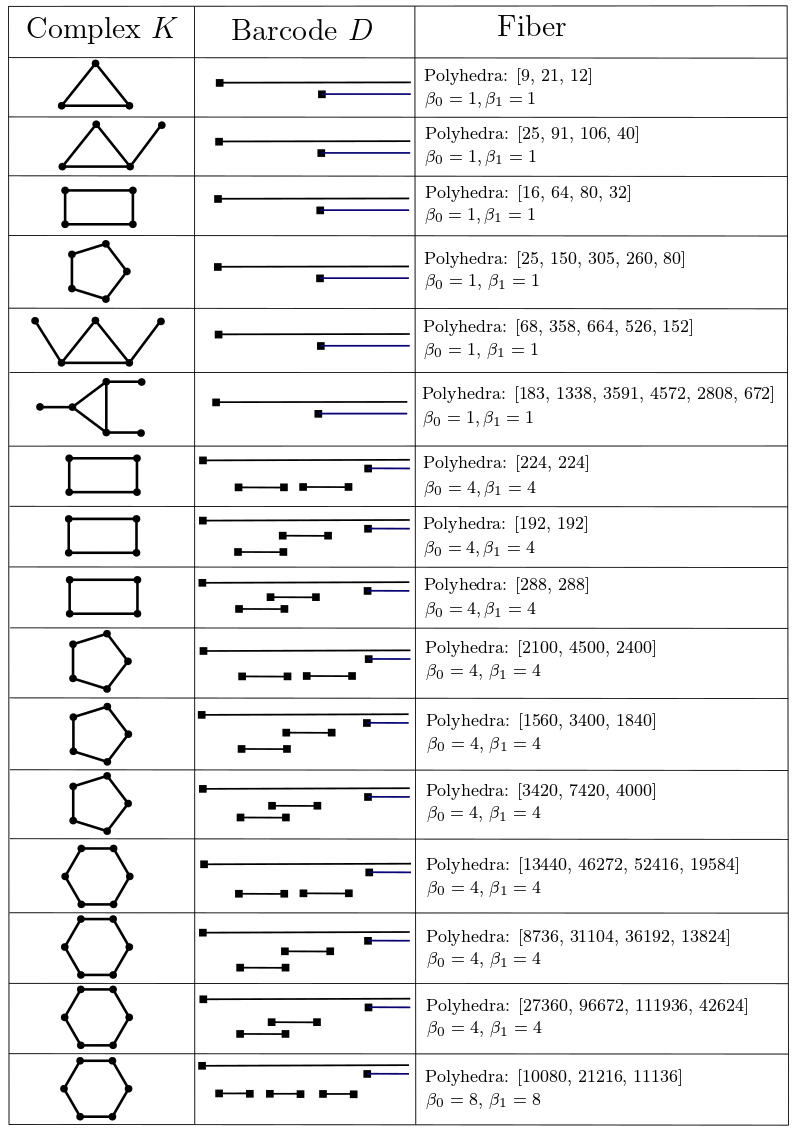}
\caption{Some statistics about fibers~$\persmap^{-1}(D)$ when~$\SComplex$ is homotopy equivalent to a circle.}
\label{fig:fiber_triangle_and_exts}
\end{figure}

When~$\SComplex$ is homotopy equivalent to a bouquet of two circles (Figure~\ref{fig:fiber_eight}), the fiber itself has trivial homology in degree higher than~$1$ and we observe cases (indicated in red) where some facets are not top-dimensional polyhedra. 
\begin{figure}[ht]
\centering
\includegraphics[width=0.8\textwidth]{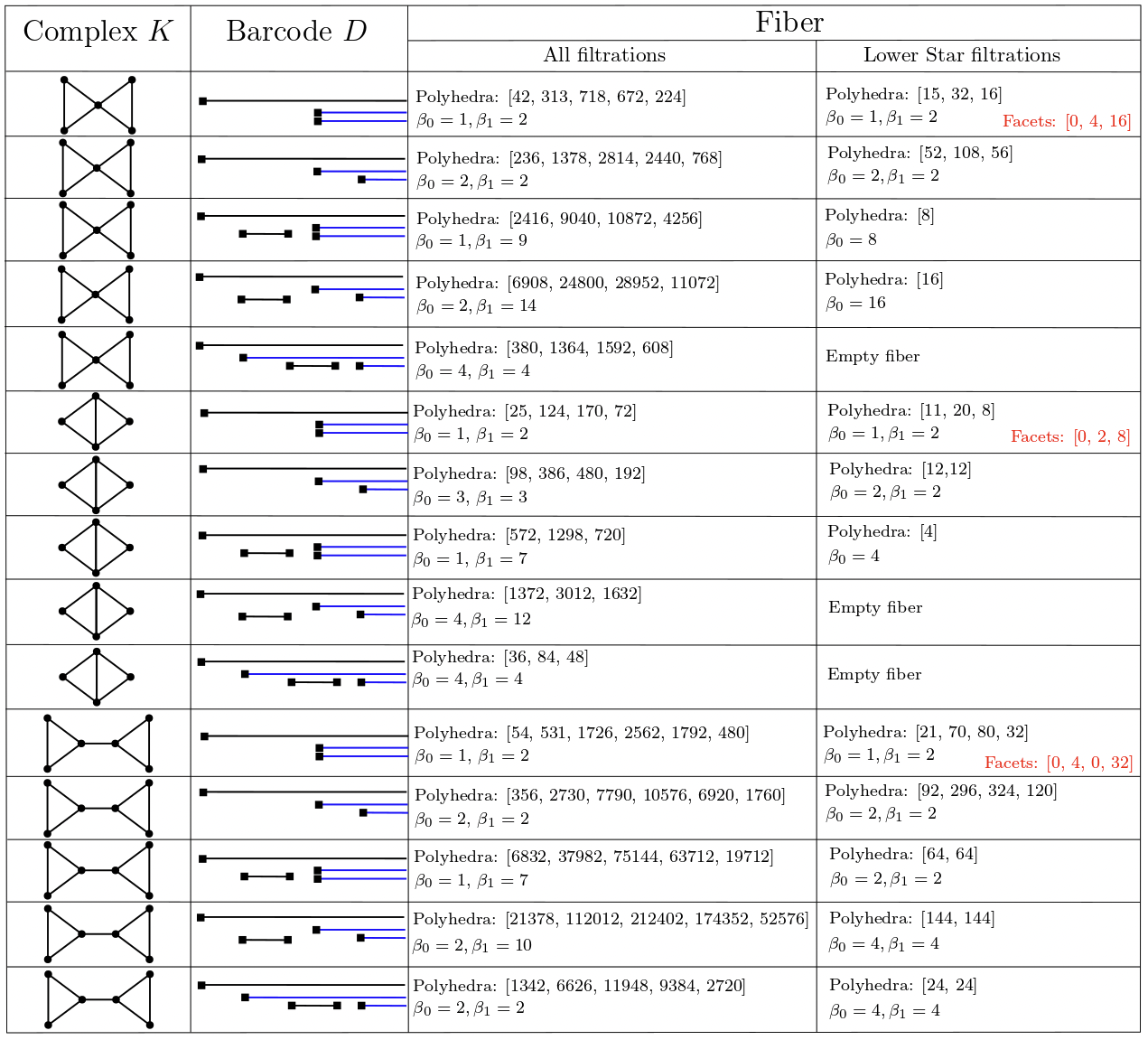}
\caption{Some statistics about fibers~$\persmap^{-1}(D)$ when~$\SComplex$ is homotopy equivalent to a bouquet of two circles.}
\label{fig:fiber_eight}
\end{figure}

In light of all the previous calculations, we can conjecture that when~$\SComplex$ is a graph the fiber~$\persmap^{-1}(D)$ has trivial homology in degrees higher than~$1$. However, when~$\SComplex$ is the $2$-skeleton of the $3$-simplex (Figure~\ref{fig:fiber_sphere_and_exts}), for some barcodes~$D$ the fiber has non-trivial degree~$3$ homology. Therefore in general the fiber~$\persmap^{-1}(D)$ may have higher non-trivial homologies than the base complex~$\SComplex$. 
\begin{figure}[ht]
\centering
\includegraphics[width=0.8\textwidth]{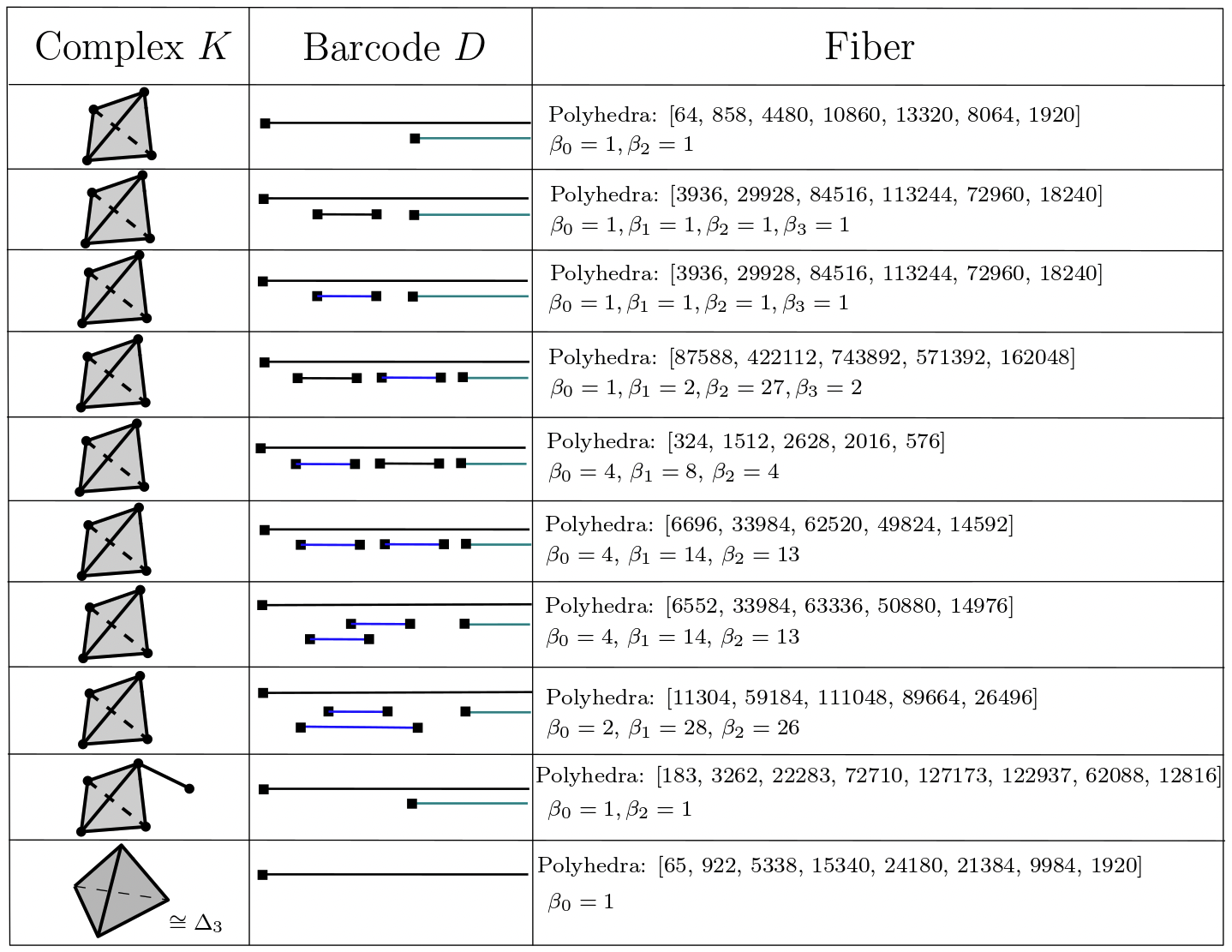}
\caption{Some statistics about fibers~$\persmap^{-1}(D)$ when~$\SComplex$ is homotopy equivalent to a $2$-sphere. }
\label{fig:fiber_sphere_and_exts}
\end{figure}

Let~$\SComplex$ be an arbitrary connected simplicial complex, and let~$D_{\SComplex}$ be the barcode with one infinite bar~$(0,+\infty)$ in degree~$0$, with no finite bars, followed by infinite bars~$(1,+\infty)$ of multiplicity~$\beta_p(\SComplex)$ in each degree~$p\geq 1$. In all the examples computed with~$\Z_2$ coefficients by our algorithm, the fiber~$\persmap^{-1}(D_\SComplex)$ and the base complex~$\SComplex$ have the same Betti numbers (with coefficients in~$Z_2$). This motivates the following conjecture.

\begin{conjecture}
\label{conjecture_fiber_homotopy_equivalent_to_K}
Let~$\SComplex$ be a simplicial complex. Then the fiber~$\persmap^{-1}(D_\SComplex)$ and~$\SComplex$ have the same Betti numbers.
\end{conjecture}

\subsection{CW Complexes}
\label{subsection_experiments_CW_complexes}

In this section~$\SComplex$ is a surface with a CW structure: the torus (Fig.\ref{fig:torus}), the Klein bottle (Fig.\ref{fig:Klein}), the real projective plane (Fig.\ref{fig:RP2}), the Möbius strip (Fig.\ref{fig:Mobius}), the cylinder (Fig.\ref{fig:Cylinder}) and the Dunce Hat (Fig.~\ref{fig:DunceHat}). Indeed from section~\ref{section_based_chain_complexes} our algorithm adapts to CW complexes and more generally to based chain complexes. This is a precious feature since simplicial triangulations of our surfaces have many simplices, hence our algorithm struggles to compute the associated fibers, while it handles cellular decompositions which are much smaller. 

For cellular triangulations that are too small (e.g. the first two decompositions of the torus in Fig.\ref{fig:torus}), fibers are not interesting. This is why we consider cellular decompositions that are not minimal and have sufficiently many simplices for the fibers to be interesting. For such fibers, the remarks of section~\ref{subsection_experiments_simplicial_complexes} about simplicial complexes apply as well. In particular, the fiber~$\persmap^{-1}(D_\SComplex)$ and the base complex~$\SComplex$ have the same Betti numbers. 

We also find novel behaviours:
\begin{itemize}
    \item For some CW complexes that are topological manifolds, such as the real projective plane and the Klein bottle, there are fibers whose facets do not consist only in top-dimensional polyhedra. In particular these fibers are not manifolds.
    \item For some CW complexes that are topological manifolds, such as the real projective plane, there are fibers whose connected components don't have the same homotopy type. This situation is detected whenever~$\beta_0(\persmap^{-1}(D))\geq 2$ and~$\beta_p(\persmap^{-1}(D))= 1$ for some~$p\geq 1$.  
    \item For some spaces like the Klein bottle and the real projective plane whose homology with coefficients in~$\Z_2$ differ from that with coefficients in~$\mathbb Q$, the fibers~$\persmap^{-1}(D)$ strongly depend on the choice field. Namely, the number of polyhedra in the fibers, the dimensions of the facets and the Betti numbers are not the same whether the persistence module associated to~$D$ is computed with coefficients in~$\Z_2$ or~$\mathbb Q$. 
    \item The dunce hat is contractible but some fibers have non-trivial $2$-dimensional homology. 
\end{itemize}

\section{Conclusion}

This work introduces and implements the first algorithm to compute the fiber~$\persmap^{-1}(D) \subseteq \R^\SComplex$.
Each fiber~$\persmap^{-1}(D)$ admits a canonical polyhedral decomposition~\cite{leygonie2021fiber}, and the output of the algorithm is a collection of polyhedra, with each polyhedron represented in computer memory as an ordered partitions of~$\SComplex$.

The proposed algorithm leverages the combinatorial structure of $\persmap^{-1}(D)$ to organize the computation as a depth-first search; this ensures that the memory requirement to run the computation (excluding the list of polyhedra returned) scales quadratically with the size of $\SComplex$, rather than exponentially, as would  naive implementations. 

In addition, we extend the polyhedral decomposition of the fiber from~\cite[Theorem 2.2]{leygonie2021fiber} to encompass not only simplicial complexes but CW complexes generally, including cubical and delta complexes.  We also incorporate variations on the notion of filter that arise naturally in applications, e.g. the lower-star filtration and Vietoris-Rips filtration.  
The proposed algorithm adapts naturally to these settings, and we include these variants in the implementation.
Indeed, this flexibility proves useful in experiments, since several computations which proved intractable on a simplicial complex $\SComplex$ due to excessive time demands later proved feasible for homeomorphic CW complexes that had fewer cells. 

This work enables the research community to study persistence fibers empirically, for the first time.  As a demonstration, we  compute the fibers of approximately 120 barcode strata, the only corpus of its kind.  The  Betti statistics of the associated polyhedral complexes suggest several numerical trends, and provide counterexamples which would be impossible to replicate by hand.  

An interesting feature of these complexes is their size.  In each of our experiments, the underlying simplicial or CW complex had fewer than 20 cells; however the associated fibers often had  hundreds of thousands of polyhedra -- in some cases, millions.   It is surprising that so many distinct solution classes should exist, given the size of $K$ and the number of conditions imposed by the persistence map.  These examples  should inform general approaches to computation in the future, and motivate the mathematical problem of formulating new, more compact representations of the fiber.

Even in cases where the fiber remains small enough to fit comfortably in computer memory, we find that challenges remain vis-a-vis overall execution time.  Most computations that are run on complexes with 15 cells or more consume hours or days; run time also depends, to a large degree, on the barcode selected.  Moreover, the overwhelming majority of internal calls to our recursive depth-first-search algorithm yield only proper faces of polyhedra that have already been computed.  This points to several natural and concrete directions either for development of new algorithms or improvement of the methods presented here.

\appendix
\section{Connection with Simple Homotopy Theory}
\label{appendix:connection_to_zeeman}
In this section we show that collapsibility of a complex~$\SComplex$, which is a combinatorial and stronger notion of contractibility, is equivalent to the fiber~$\persmap^{-1}(D)$ over a well-chosen barcode~$D$ being non-empty. In particular we can use our algorithm for computing~$\persmap^{-1}(D)$ to determine whether~$\SComplex$ is collapsible.

Given~$\tau,\simplex$ two simplices, $\tau \subseteq \simplex$ and $\dim \simplex=\dim \tau +1$, such that~$\sigma$ is a maximal face of~$\SComplex$ and no other simplex contains~$\tau$, we say that~$\tau$ is a {\em free face}. The operation of removing~$\tau,\simplex$ is called an {\em elementary collapse}, and if~$L:=\SComplex \setminus \{\tau\subseteq \simplex\}$ is the resulting complex we write~$\SComplex \searrow L$. Finally~$\SComplex$ is said to be {\em collapsible} if there is a sequence of elementary collapses from~$\SComplex$ to one of its vertices:
\[K=L_n\searrow L_{n-1} \searrow L_{n-2} \searrow \cdots \searrow L_1 \searrow L_0=\{v\}.\]
Collapsibility implies contractibility but the reverse is false: the dunce hat and the house with two rooms are instances of contractible $2$-complexes that are not collapsible. However we have the following well-known Zeeman's conjecture, appropriately phrased in~\cite{adiprasito2017subdivisions}  for simplicial complexes:
\begin{conjecture}[Zeeman~\cite{zeeman1963dunce}]
Let~$\SComplex$ be a contractible $2$-complex. Then after taking finitely many barycentric subdivisions the product~$\SComplex \times \I$ is collapsible. 
\end{conjecture}
This conjecture remains open and implies the 3-dimensional conjecture~\cite{zeeman1963dunce}.

Next we bridge the question of the collapsibility of a complex~$\SComplex$ to the fiber of~$\persmap$ over barcodes~$D$ that are {\em elementary}: those have~$1$ infinite bar~$(b_0,\infty)$ in dimension~$0$ followed by~$\frac{\sharp K-1}{2}$ non-overlapping intervals~$(b_i,d_i)$, that is:
\[b_0<b_1<d_1<b_2<d_2<\cdots <b_i<d_i< \cdots < b_{\frac{\sharp K-1}{2}}<d_\frac{\sharp K-1}{2}\]
\begin{proposition}
\label{proposition_bridge_collapsibility_elementary_barcode}
Let~$\SComplex$ be a contractible complex. Then~$\SComplex$ is collapsible if and only if there exists an elementary barcode~$D$ with nonempty fiber, i.e.~$\persmap^{-1}(D)\neq \emptyset$.

\end{proposition}
\begin{proof}
If~$\SComplex$ is collapsible let~$K=L_n\searrow L_{n-1} \searrow L_{n-2} \searrow \cdots \searrow L_1 \searrow L_0=\{v\}$ be a sequence of elementary collapses, with notations~$n=\frac{\sharp K-1}{2}$ and $L_{i+1}=L_i\cup \{\tau_i\subseteq \simplex_i\}$, and define a filter~$f$ by~$f(v):=0$,~$f(\tau_i):=2i+1$ and~$f(\simplex_i):=2i+2$. Then the barcode of~$f$ is clearly elementary by definition of an elementary collapse.  

Conversely, let~$D=\{(b_0,\infty)\}\cup \{(b_i,d_i)\}_{1\leq i \leq \frac{\sharp \SComplex -1}{2}}$ be an elementary barcode and~$f$ a filter in the fiber, i.e.~$\persmap(f)=D$. Since~$D$ has exactly~$\sharp \SComplex$ distinct endpoint values,~$f$ establishes a bijection~$v\mapsto b_0$, $(\tau_i,\simplex_i)\mapsto (b_i,d_i)$, from simplices of~$\SComplex$ to these endpoints. In particular~$\tau_n$ and then~$\simplex_n$ are the last two simplices to enter the sub-level set filtration of~$f$, so that~$\simplex_n$ is a maximal face and no other simplex can contain~$\tau_n$. However~$\simplex_n$ itself contains~$\tau_n$ because the~$\dim \tau_n= \dim \simplex_n -1= \dim (b_n,d_n)$-cycle created by~$\tau_n$ becomes a boundary when adding~$\simplex_n$ in the filtration. Thus removing~$\tau_n$ and~$\simplex_n$ defines an elementary collapse, and we conclude by induction. 
\end{proof}

\section{Adaptation for persistent (relative) (co)homology}
\label{sec:dualities}

\newcommand{\chaincxAtTime}[1]{L_{#1}}
\newcommand{\phmod}{\mathsf{H}}
\renewcommand{\H}{H}

In addition to the homology functor, the relative homology, cohomology, and relative cohomology functors  engender distinct persistence modules of their own, each of which determines a barcode and thus a new persistence map.  We claim that the procedure described to compute the fiber of the persistent homology map in this work also suffices to compute the fibers of these other maps.

Let $(\chaincx, \basis)$ be a based, finite-dimensional, $\field$-linear chain complex equipped with a filter $\filta: \basis \to \I $ that surjects onto a finite subset of the unit interval~$\I$, denoted $\critvals = \{\critval_0 < \cdots < \critval_\ncritval \}$,  where $\critval_0 = 0$ and $\critval_m < 1$. Write $\chaincxAtTime{t}$ for the linear span of $\{ e \in \basis : \filta(e) \le t\}$, which forms a subcomplex of $\chaincx$ by hypothesis.

From these data we can construct four distinct sequences of vector spaces and homomorphisms,  induced by either inclusion or quotient:

\begin{equation*}
\xymatrix@R-2pc{
    \phmod_*( \chaincxAtTime{}): && \H_*(\chaincxAtTime{\critval_0})  \ar[r] & 
    \quad \cdots \quad \ar[r] & \H_*(\chaincxAtTime{\critval_\ncritval })  \ar[r] & \H_*(\chaincxAtTime{1})
    \\
    \phmod^*( \chaincxAtTime{}): && \H^*(\chaincxAtTime{\critval_0}) & \ar[l]
    \quad \cdots \quad & \ar[l] \H^*(\chaincxAtTime{\critval_\ncritval })  & \ar[l] \H^*(\chaincxAtTime{1})  
    \\
    \phmod_*( \chaincxAtTime{1}, \chaincxAtTime{}): &
    \H_*(\chaincxAtTime{1})  \ar[r]    
    & \H_*(\chaincxAtTime{1}, \chaincxAtTime{\critval_0})  \ar[r] & 
    \quad \cdots \quad \ar[r] & \H_*(\chaincxAtTime{1}, \chaincxAtTime{\critval_\ncritval })   
    \\
    \phmod^*( \chaincxAtTime{1}, \chaincxAtTime{}): &
    \H^*(\chaincxAtTime{1})  
    & \ar[l]    \H^*(\chaincxAtTime{1}, \chaincxAtTime{\critval_0})  
    & 
    \ar[l] 
    \quad \cdots \quad 
    & 
    \ar[l] 
    \H^*(\chaincxAtTime{1}, \chaincxAtTime{\critval_\ncritval })       
}
\end{equation*}

We refer to these as the  homology, cohomology, relative homology, and relative cohomology persistence modules, respectively.

A classic result of \cite{de2011dualities} states that the barcode for $\phmod_*(\chaincxAtTime{})$ uniquely determines the barcodes for $\phmod^*(\chaincxAtTime{}), \phmod_*(\chaincxAtTime{1}, \chaincxAtTime{}),$ and $\phmod^*(\chaincxAtTime{1}, \chaincxAtTime{})$.

To compute the fiber of one of these other maps, therefore, one  must simply convert the barcode to the associated PH barcode and apply any algorithm that is specialized to compute fibers for $\phmod_*( \chaincxAtTime{})$.  Barcodes in persistent homology can be converted into barcodes for the other three standard persistence modules as follows\footnote{For full details see \cite{de2011dualities}.  The authors of that work stipulate that $(\chaincx, \basis)$ be the  chain complex of a filtered CW complex, equipped with the standard basis of cells;  however no proofs make use of this added restriction on $(\chaincx, \basis)$, and the results are easily seen to hold for arbitrary based filtered complexes.}
    \begin{enumerate}
        \item $\phmod_*(\chaincxAtTime{})$ from  $\phmod^*(\chaincxAtTime{})$: no change
        \item $\phmod_*(\chaincxAtTime{})$ from  $\phmod_*(\chaincxAtTime{1}, \chaincxAtTime{})$ or $\phmod^*(\chaincxAtTime{1}, \chaincxAtTime{})$: subtract 1 from the homology degree of each finite bar; replace each infinite bar of form  $(-\infty, a)$ with $[a,\infty)$, leaving degree unchanged        
    \end{enumerate}

\section{Additional computations}
\begin{figure}[ht]
\centering
\includegraphics[width=0.8\textwidth]{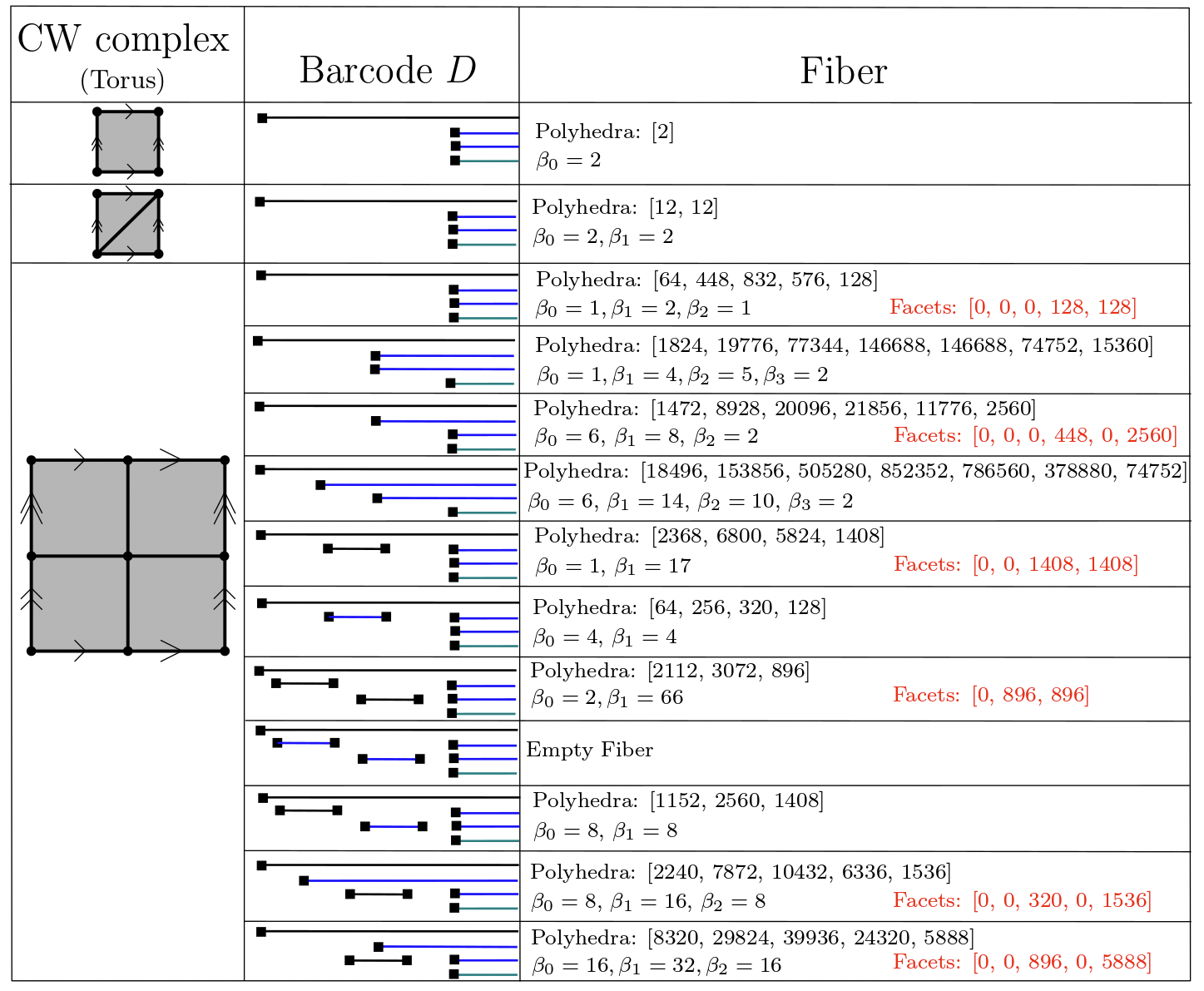}
\caption{Some statistics about fibers~$\persmap^{-1}(D)$ when~$\SComplex$ is a CW decomposition of the torus.}
\label{fig:torus}
\end{figure}

\begin{figure}[ht]
\centering
\includegraphics[width=\textwidth]{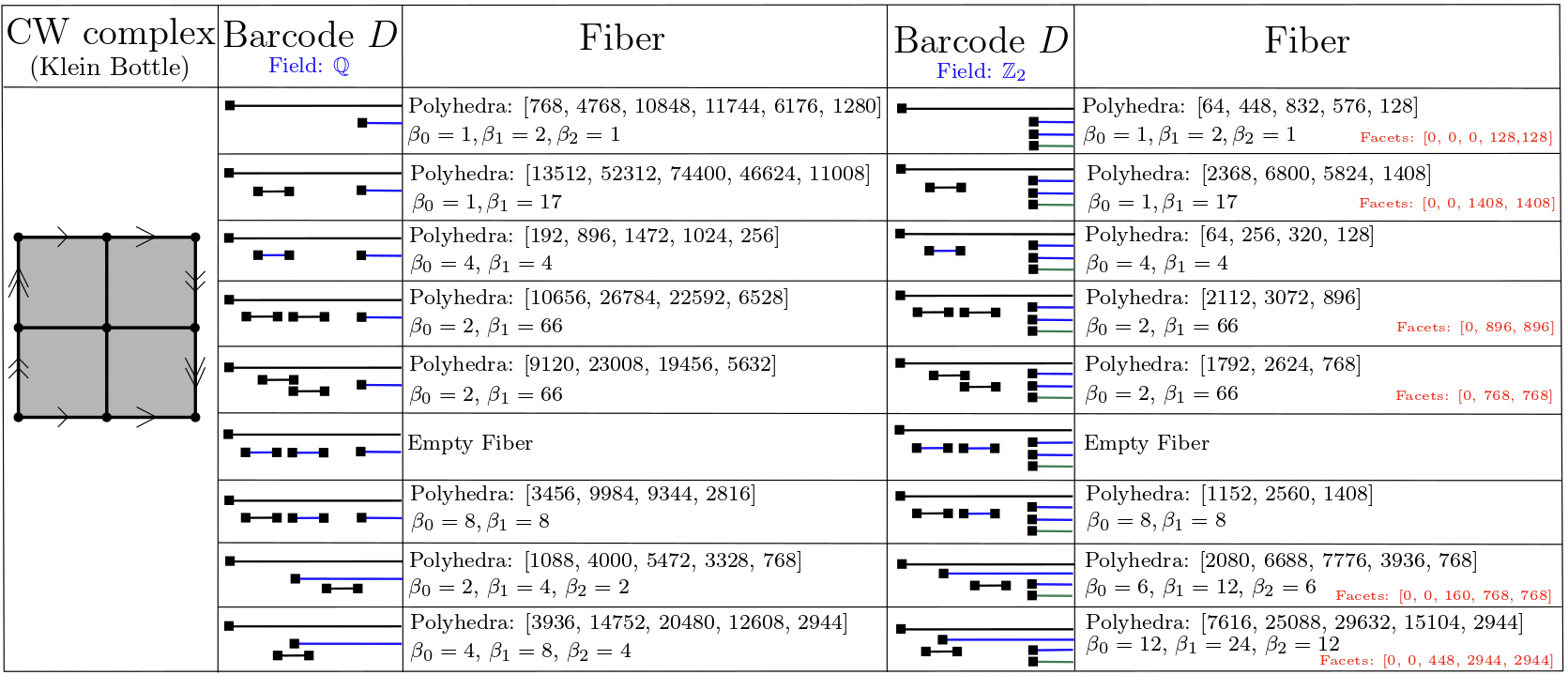}
\caption{Some statistics about fibers~$\persmap^{-1}(D)$ when~$\SComplex$ is a CW decomposition of the Klein bottle, and when persistence modules are computed with coefficients in the field~$\mathbb Q$ or~$\mathbb Z_2$. }
\label{fig:Klein}
\end{figure}

\begin{figure}[ht]
\centering
   \begin{subfigure}{0.8\textwidth}
   \includegraphics[width=0.8\textwidth]{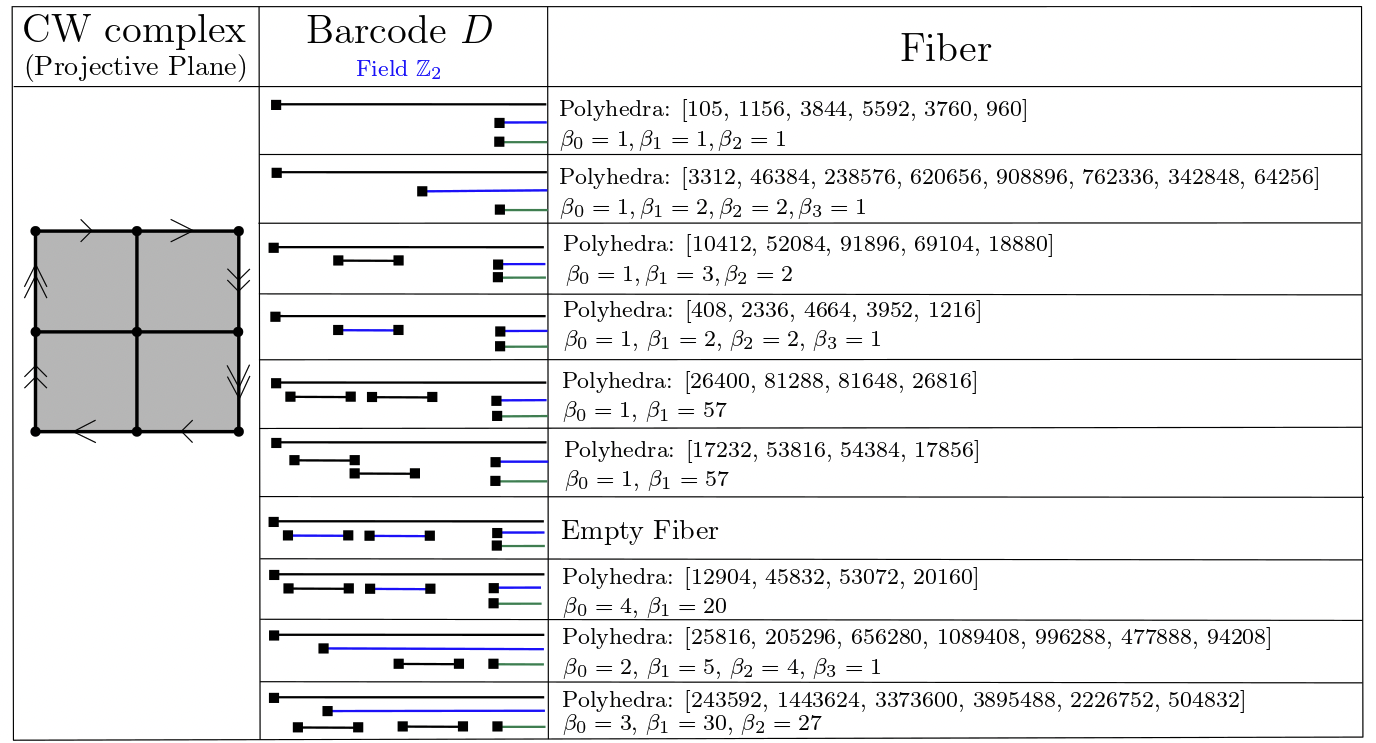}
   \label{fig:Ng1} 
\end{subfigure}
\begin{subfigure}{0.8\textwidth}
   \includegraphics[width=0.8\textwidth]{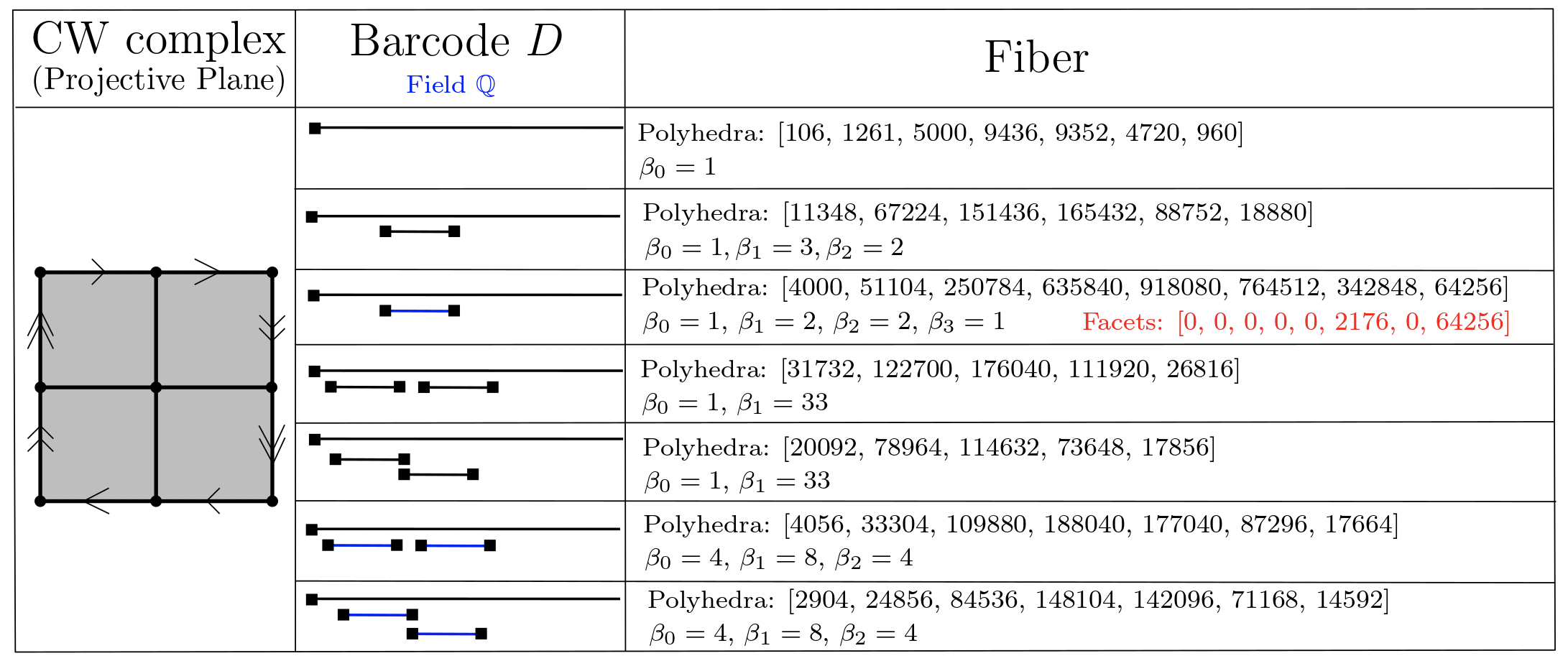}
\end{subfigure}
\caption{Some statistics about fibers~$\persmap^{-1}(D)$ when~$\SComplex$ is a CW decomposition of the real projective plane, and when persistence modules are computed with coefficients in the field~$\mathbb Q$ or~$\mathbb Z_2$. }
 \label{fig:RP2}
\end{figure} 
\begin{figure}[ht]
\centering
\includegraphics[width=0.55\textwidth]{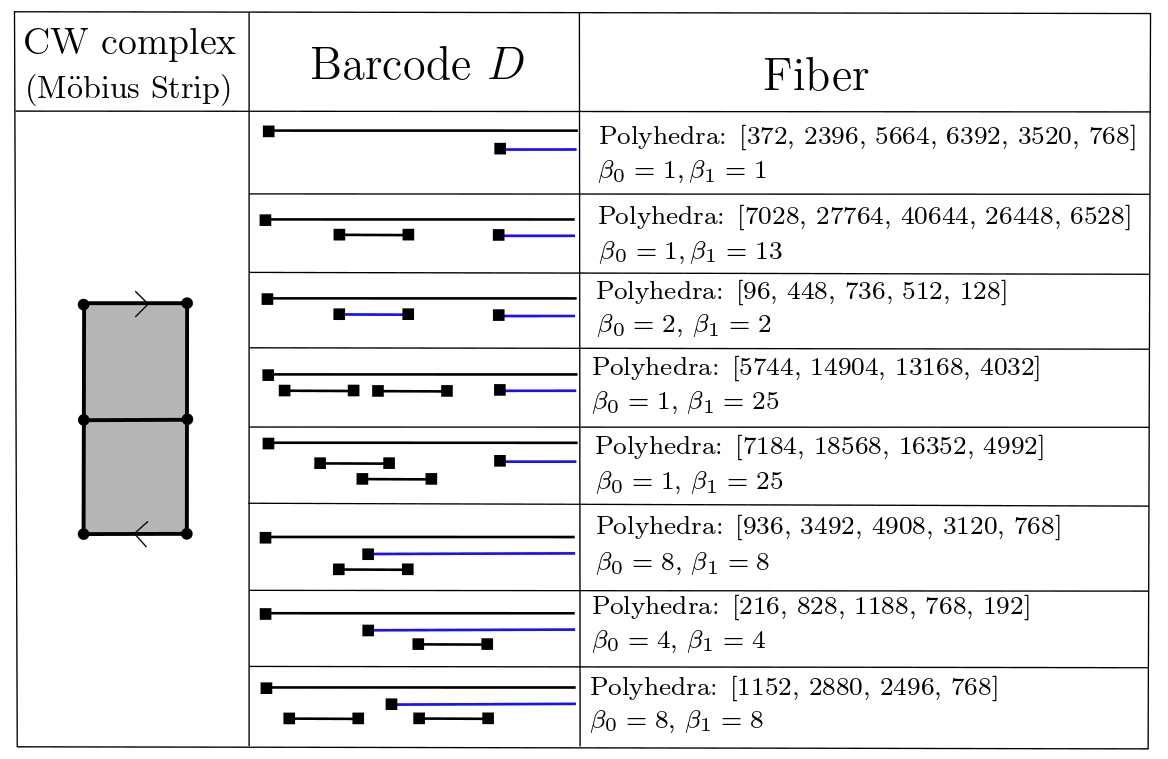}
\caption{Some statistics about fibers~$\persmap^{-1}(D)$ when~$\SComplex$ is a CW decomposition of the Möbius strip. }
\label{fig:Mobius}
\end{figure}
\begin{figure}[ht]
\centering
\includegraphics[width=0.55\textwidth]{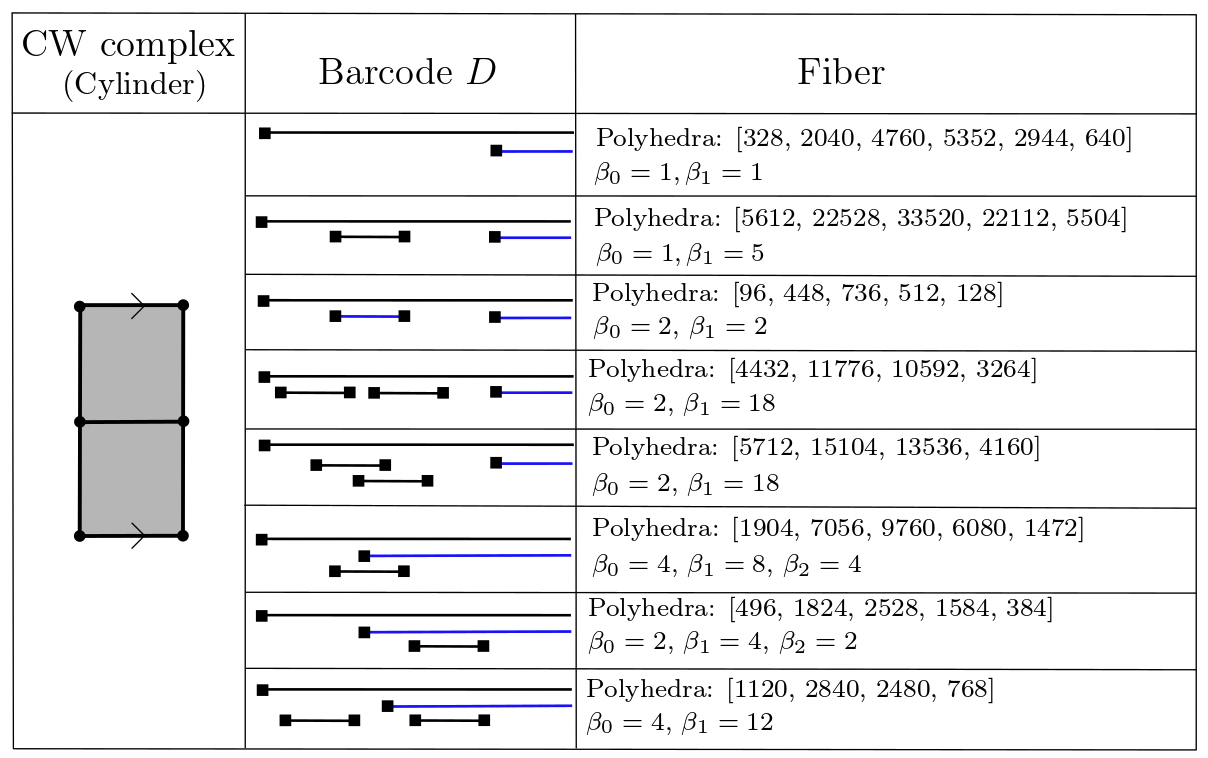}
\caption{Some statistics about fibers~$\persmap^{-1}(D)$ when~$\SComplex$ is a CW decomposition of the cylinder. }
\label{fig:Cylinder}
\end{figure}

\begin{figure}[ht]
\centering
\includegraphics[width=0.7\textwidth]{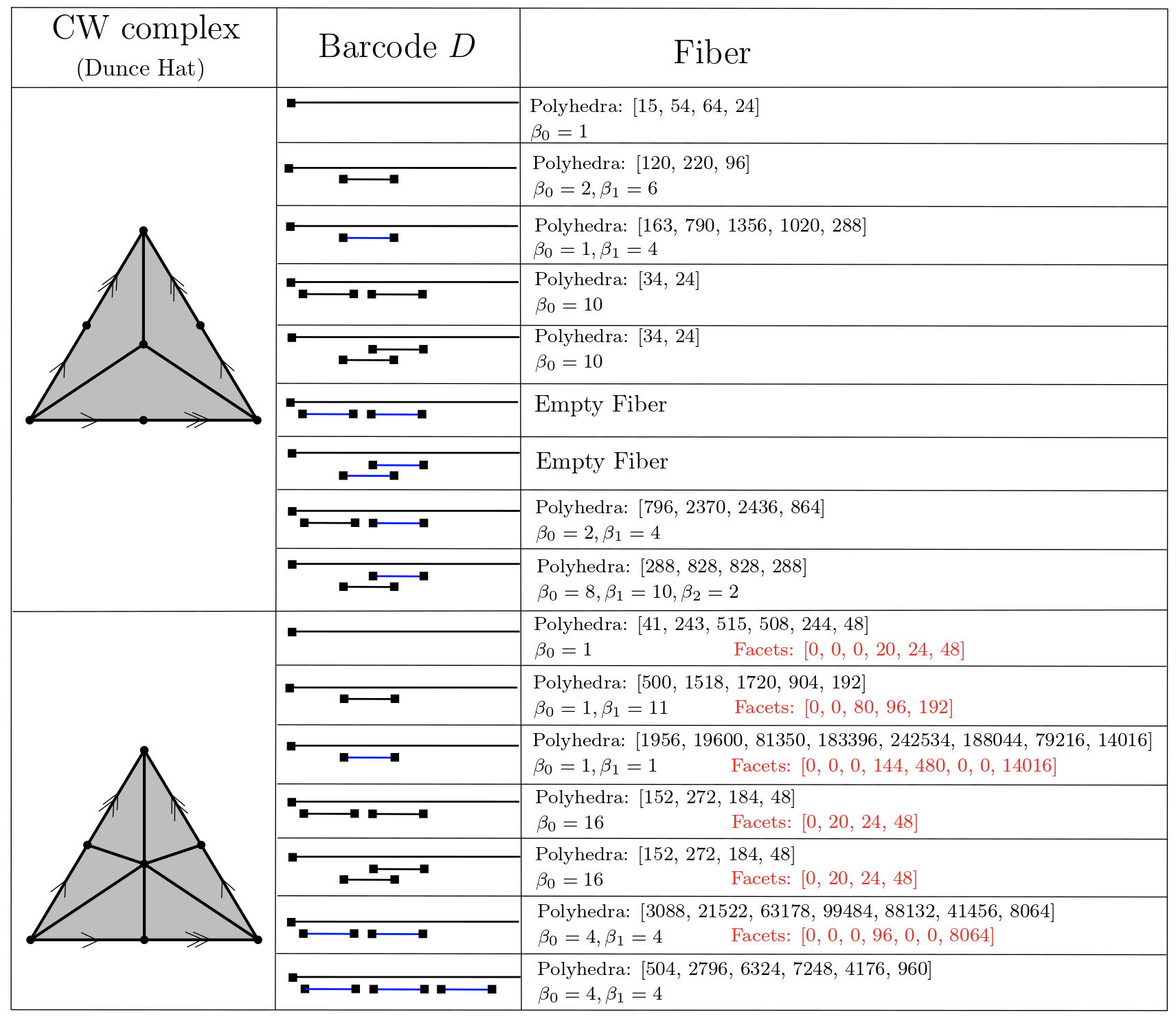}
\caption{Some statistics about fibers~$\persmap^{-1}(D)$ when~$\SComplex$ is a CW decomposition of the Dunce Hat.}
\label{fig:DunceHat}
\end{figure}

\bibliographystyle{plain}
\bibliography{reference}
\end{document}